\documentclass[12pt,draftcls,onecolumn]{IEEEtran}  

\IEEEoverridecommandlockouts                              
\usepackage{dsfont}
\usepackage{graphicx}
\usepackage{epsfig}
\usepackage{subfig}
\usepackage{times}
\usepackage{amsmath}
\usepackage{amssymb}
\usepackage{cite}
\usepackage{color}
\usepackage{arydshln}
\usepackage[left, pagewise]{lineno}
\usepackage{tikz}
\usetikzlibrary{arrows}

\newtheorem{definition}{Definition}[section]

\newtheorem{remark}{Remark}[section]
\newtheorem{lemma}{Lemma}[section]
\newtheorem{theorem}{Theorem}[section]

\newtheorem{assumption}{Assumption}[section]

\graphicspath{{figure/}}
\def\QED{~\rule[-1pt]{5pt}{5pt}\par\medskip}
\newenvironment{proof}{{\it  Proof: \ }}{ \hfill \QED}

\newenvironment{proofof}{{\it \quad Proof of  }}{ \hfill \QED}

\graphicspath{{figure/}}

\newcommand{\sign}{\operatornamewithlimits{sign}}

\newcommand{\diag}{\operatornamewithlimits{diag}}
\newcommand{\blkdiag}{\operatornamewithlimits{blkdiag}}

\newcommand{\argmin}{\operatornamewithlimits{argmin}}

\newcommand{\spand}{\operatornamewithlimits{span}}
\newcommand{\Ker}{\operatornamewithlimits{Ker}}

\title{\LARGE Consensus with Output Saturations
}
\author{\normalsize Young-Hun Lim and Hyo-Sung Ahn
\thanks{The authors are with School of Mechatronics,
Gwangju Institute of Science and Technology, Gwangju, Korea.
E-mail: {hoonnim@gist.ac.kr}; {hyosung@gist.ac.kr.}}%
}

\begin{document}

\maketitle 
\thispagestyle{empty}
\pagestyle{empty}


\begin{abstract}
This paper consider a standard consensus algorithm under output saturations.
In the presence of output saturations, global consensus can not be realized due to the existence of stable, unachievable equilibrium points for the consensus.
Therefore, this paper investigates necessary and sufficient initial conditions for the achievement of consensus, that is an exact domain of attraction.
Specifically, this paper considers singe-integrator agents with both fixed and time-varying undirected graphs, as well as double-integrator agents with fixed undirected graph.
Then, we derive that the consensus will be achieved if and only if the average of the initial states (only velocities for double-integrator agents with homogeneous saturation levels for the outputs) is within the minimum saturation level.
An extension to the case of fixed directed graph is also provided in which an weighted average is required to be within the minimum saturation limit.
\end{abstract}

\section{INTRODUCTION}
For the last one decade, the consensus problem has been attracted a lot of attention due to wide applications such as flocking, sensor networks, unmanned air vehicle (UAV) formations, etc (see, e.g., \cite{Olfati-Saber:2004,Olfati-Saber:2007,Ren:2007} and the references therein). 
In this problem setup, each agent measures its own state, and exchanges this information with its neighbors such that the states of all agents converge to a certain value.
Consider a group of $N$ single-integrator modeled agents, and let $x_{i}, y_{i} \in \mathbf{R}$ be the state and the measured output of agent $i$.
Then, a standard consensus algorithm takes the following form \cite{Olfati-Saber:2004}:
\begin{align} \label{int_sys}
\dot{x}_{i} = \sum_{j=1}^{N} \alpha_{ij} (t) ( y_{j} - y_{i} ), \,~~ i \in \mathcal{V} := \{1,...,N\},
\end{align}
and then, the overall networked agent has the form
\begin{align} \label{int_sys1}
\dot{x} = - L (t) y.
\end{align}
For an undirected graph, the null space of the Laplacian matrix $L(t)$ is $\spand \{ \mathbf{1} \}$. Therefore, an equilibrium of (\ref{int_sys1}) is the state in the form $y^{*} = C \mathbf{1}$, $C \in \mathbf{R}$.
If each agent can measure the exact state, i.e., $y_{i} = x_{i}$, $\forall i \in \mathcal{V}$,
then, when the fixed graph is connected \cite{Olfati-Saber:2004,Olfati-Saber:2007,Ren:2007} or the time-varying graph is integrally connected over $[0, \infty)$ \cite{Cao:2011}, $y^{*} = x^{*} = C \mathbf{1}$ is an unique equilibrium that implies the agents (\ref{int_sys}) achieve the consensus.

Meanwhile, in real applications, the usage of measurement units may lead to nonlinearities over the network.
For example, due to digital communication channels or digital sensors, the consensus problems under quantization effects have been studied for a fixed graph in \cite{Ceragioli:2011} and for a time-varying graph in \cite{Persis:2012}.
The consensus has been derived by utilizing some properties of the Laplacian \cite{Ceragioli:2011} and the integral graph \cite{Persis:2012}.
In \cite{Hui:2008,Liu:2009,Nosrati:2012}, the consensus problems have been studied for more general nonlinearities with (strictly) increasing or decreasing conditions.
In the above results, the nonlinearities were assumed to be unbounded.

On the other hands, there exists bounded nonlinearity, called output saturation, due to range limitations of the measurement units.
Therefore, the control problem of systems subject to output saturation has been widely studied by several authors.
Global and semi-global stabilization problems have been studied in \cite{Kreisselmeier:1996,Grip:2010,Kaliora:2004} and \cite{Lin:2001}, respectively.
In \cite{Turner:2009}, a dynamic anti-windup strategy has been discussed.
While the stabilization under output saturations has been addressed in much detail, the consensus problem has received fewer results \cite{Liu:2012,Lim:2014}.
Note that, in \cite{Liu:2009,Lim:2014}, it was pointed out that for the bounded nonlinearities, the consensus may not be realized due to the existence of  stable, unachievable equilibrium points for consensus (see \textit{Remark 3} in \cite{Liu:2009}).
Let us consider a simple example when $y_{i} = \mbox{sat} ( x_{i} )$ with $\mbox{sat} ( \cdot ) = \sign ( \cdot ) \max \{ | \cdot |, 1 \}$.
In this case, an equilibrium of (\ref{int_sys1}) is the state in the form $\mbox{sat} ( x^{*} ) = C \mathbf{1}$,
Thus, the set of equilibria of (\ref{int_sys1}) can be divided into two groups as
$\Omega_{a} := \{ x \in \mathbf{R}^{N} : x = C \mathbf{1}, | x | \le \mathbf{1} \}$ and $\Omega_{u} = \Omega_{u^{+}} \cup \Omega_{u^{-}}$, where  $\Omega_{u^{+}} := \{ x \in \mathbf{R}^{N} : x \ge  \mathbf{1} ,  x  \neq  \mathbf{1} \}$ and $\Omega_{u^{-}} := \{ x \in \mathbf{R}^{N} : x \le - \mathbf{1} ,  x  \neq  - \mathbf{1} \}$.
It is clear that $\Omega_{a}$ is the set of achievable equilibrium, i.e., $x^{*} \in \Omega_{a}$ implies that the consensus is achieved, but $\Omega_{u}$ may not (see, Section~\ref{simulation}).
Therefore, \cite{Liu:2012,Lim:2014} have developed the consensus algorithms under the bounded constraints.
Specifically, in \cite{Liu:2012}, the discarded consensus algorithm, which discards the state of a neighbor if the state is outside its constraint, was proposed.
In \cite{Lim:2014}, the output feedback based leader-following consensus algorithm was studied.

Note that, under the standard consensus algorithm, the agents converge to the average value.
However, as mentioned above, the consensus with the standard setup may not be achieved under output saturations.
Although some results have been available for the consensus problem under output saturations, an analytic result has not been achieved.
Therefore, this paper investigates conditions for achieving consensus under output saturations.   

We consider the dynamics of each agent as a single-integrator, and both fixed and time-varying undirected graphs.
Moreover, we consider homogeneous and heterogeneous saturation levels, in which the agents have identical and different saturation levels, respectively.
Then, we first analyze the consensus under the fixed and connected graph.
By utilizing an integral Lyapunov function, we investigate necessary and sufficient conditions for achieving the consensus, that is an exact domain of attraction.
We next consider the consensus under the time-varying graph topology with an integrally connected condition, which is the necessary and sufficient graph condition for achieving the consensus.
We analyze an attractivity of equilibrium, and then by investigating conditions for the achievable equilibrium, we derive the necessary and sufficient conditions.
Moreover, we extend the results to the cases of  double-integrator modeled agents as well as the fixed, directed graph.

Sequentially, the main contributions of this paper are as follows.
First, under the standard consensus algorithm, we prove an asymptotic convergence of agents with output saturations.
We consider general saturation levels and graph topology.
The analysis techniques of this paper rely on the strictly increasing property of the saturation function within its bounds.
Thus, the analysis can be easily extended to any bounded nonlinearities, which are strictly increasing within its bounds.
Second, we investigate some properties of the set of equilibria.
Then, necessary and sufficient initial conditions for achieving the consensus are obtained, that is an exact domain of attraction.
Third, the analytic results are extended to the cases of double-integrator modeled agents as well as fixed and directed graph cases.

The remainder of this paper is organized as follows. In Section~\ref{preliminaries}, some basic definitions and notations are reviewed, and the problem statement and preliminaries are presented.
In Section~\ref{sec_fixed} and Section~\ref{sec_time-varying}, necessary and sufficient conditions are derived for fixed and time-varying graphs, respectively.
In Section~\ref{extension}, we further consider the case of double-integrator modeled agents, and fixed and directed graph.
In Section~\ref{simulation}, numerical examples are presented.
Then, the conclusions and suggestions for future work are presented in Section~\ref{conclusion}, and  some of the proofs are given in Appendix.

\section{Problem statement and preliminaries}	\label{preliminaries}


\subsection{Graph Theory}

A (fixed) graph $\mathcal{G}$ is defined as three-tuple $(\mathcal{V}, \mathcal{E}, \mathcal{A} )$, where $\mathcal{V}$ denotes the set of nodes, $\mathcal{E} \subseteq \mathcal{V} \times \mathcal{V}$ denotes the set of edges, and $\mathcal{A} = [ \alpha_{ij} ] \in \mathbf{R}^{N \times N}$, where $\alpha_{ij}$ is the weight assigned to edge $(i,j)$, denotes the underlying weighted adjacency matrix defined as $\alpha_{ij} > 0$ if $(i,j) \in \mathcal{E}$, and $\alpha_{ij} = 0$ otherwise.
The Laplacian matrix of the graph is defined as $L = \mathcal{D} - \mathcal{A}$, where $\mathcal{D} = \diag (\mathcal{A} \mathbf{1}_{N} ) \in \mathbf{R}^{N \times N}$

A graph $\mathcal{G} = (\mathcal{V}, \mathcal{E}, \mathcal{A})$ is said to be undirected if $(i,j) \in \mathcal{E}$, then $(j, i) \in \mathcal{E}$, that is $\alpha_{ij} = \alpha_{ji}$, $\forall i,j \in \mathcal{V}$,
otherwise it is termed a directed graph.
For the undirected graph, the adjacency matrix is symmetric, i.e., $\mathcal{A}^{T} = \mathcal{A}$, and thus $L$ is positive semidefinite real symmetric matrix, so all eigenvalues of $L$ are non-negative real.
For the directed graph, $L$ needs no longer to be symmetric, but the eigenvalues of $L$ have non-negative real part.
A directed path is a sequence of edges in the directed graph of the form $(i_{1}, i_{2}), (i_{2}, i_{3}), ...$.
An undirected path in the undirected graph is defined analogously.

\begin{definition} (Connectivity of fixed graph) \label{def_connec_fixed}
A directed graph $\mathcal{G}$ is said to be strongly connected if there exists a directed path between any two distinct nodes.
An undirected graph $\mathcal{G}$ is said to be connected if there exists an undirected path between any two distinct nodes.
\end{definition}

For an undirected graph, $0$ is a simple eigenvalue of $L$ if and only if the undirected graph is connected.
For a directed graph, $0$ is a simple eigenvalue of $L$ if the directed graph is strongly connected.

A graph is said to be time-varying if it changes over time $t$, and denoted by $\mathcal{G} (t) = ( \mathcal{V}, \mathcal{E} (t) , \mathcal{A} (t) )$.

\begin{definition} \cite{Cao:2011} (Integral graph) \label{def_integral_graph}
Given a time-varying graph $\mathcal{G} (t) = ( \mathcal{V}, \mathcal{E} (t), \mathcal{A} (t) )$, the integral graph of $\mathcal{G}(t)$ on $[0, \infty)$ is a constant graph $\bar{\mathcal{G}}_{[0,\infty)} := ( \mathcal{V}, \bar{\mathcal{E}}, \bar{\mathcal{A}})$, where $\mathcal{V}$ is the same node set of $\mathcal{G} (t)$, and $\bar{\mathcal{A}} = [ \bar{\alpha}_{ij}] \in \mathbf{R}^{N \times N}$ is defined by $\bar{\alpha}_{ij} = 1$ if $\int_{0}^{\infty} \alpha_{ij} (t) dt = \infty$, and $\bar{\alpha}_{ij} = 0$ otherwise.
\end{definition}

\begin{definition} \cite{Cao:2011} (Integral connectivity of time-varying undirected graph) \label{def_integral_connec}
A time-varying undirected graph $\mathcal{G} (t)$ is said to be integrally connected over $[0, \infty)$ if its integral graph $\bar{\mathcal{G}}_{[0,\infty)}$ is connected.
\end{definition}

\begin{remark} \cite{Cao:2011}
If a time-varying undirected graph $\mathcal{G} (t)$ is integrally connected over $[0, \infty)$, then there exists a time interval $0 = t_{0} < t_{1} < \cdots < t_{k} < \cdots$ such that $\int_{t_{k-1}}^{t_{k}} L (t) dt$ is connected $\forall k \ge 0$.
\end{remark}

Analogous criterion referred as the ``$\delta$-connected graph'' was studied in \cite{Shi:2013}.
An edge $(i,j)$ is said to be a $\delta$-edge of $\mathcal{G} (t)$ on time interval $[t_{k-1}, t_{k} )$ if $\int_{t_{k-1}}^{t_{k}} \alpha_{ij} (t) \ge \delta$.

\begin{definition} \cite{Shi:2013} ($\delta$-connected graph) \label{def_delta_connec}
A time-varying graph $\mathcal{G} (t)$ is said to be uniformly $\delta$-connected if there exists a constant $T >0$ such that for any $t \ge 0$, the $\delta$-edges of $\mathcal{G} (t)$ on time interval $[t, t+T)$ form a connected graph.
If the $\delta$-edges on time interval $[t, \infty)$ form a connected graph, then the graph $\mathcal{G} (t)$ is said to be infinitely $\delta$-connected.
\end{definition}

\subsection{Problem Formulation}

We consider a group of $N$ single-integrator modeled agents under output saturations, and the following standard consensus algorithm:
\begin{align} \label{ch4_sys}
\dot{x}_{i} =& \sum_{j=1}^{N} \alpha_{ij} (t) ( y_{j} - y_{i} ), \,~~ i \in \mathcal{V} := \{ 1,...,N \}, \nonumber\\
y_{i} =& \mbox{sat}_{i} ( x_{i} ),
\end{align}
where $x_{i} , y_{i} \in \mathbf{R}$ are the state and the measured output of the agent $i$, and the saturation function is defined as
\begin{align} \label{ch4_saturation}
\mbox{sat}_{i} ( x_{i} ) = \sign ( x_{i} ) \max \{ | x_{i} | , s_{i} \}, \,~~ s_{i} > 0,
\end{align}
where $s_{i}$ represents the saturation level, and we use $\mbox{sat}_{i} ( x_{i} ) = \mbox{sat} (x_{i})$ for $s_{i} = s$, $\forall i \in \mathcal{V}$.
Then, we say that the agents are homogeneous if $s_{i} = s$, $\forall i \in \mathcal{V}$, and heterogeneous otherwise.

This paper studies the consensus problem for the $N$ agents with output saturations (\ref{ch4_sys}).
\begin{definition} \label{def_consensus}
The consensus is said to be achieved for the group of $N$ agents if $\lim_{t \rightarrow \infty} x_{i}  = C$, $\forall i \in \mathcal{V}$, where $C \in \mathbf{R}$ is called the group decision value.
\end{definition}

\begin{lemma} \label{ch4_lem_average_undirected}
Consider the group of $N$ agents (\ref{ch4_sys}), and suppose the graph is undirected.
Then, the average of all agent states $\frac{1}{N} \sum_{i=1}^{N} x_{i} (t)$ is invariant, $\forall t \ge t_{0}$.
\end{lemma}
\begin{proof}
The time derivative of the average value is given by $\frac{1}{N} \sum_{i=1}^{N} \dot{x}_{i} (t) = \frac{1}{N} \mathbf{1}^{T} \dot{x} = - \frac{1}{N} \mathbf{1}^{T} L (t) y = 0$. Therefore, the average value is preserved, $\forall t \ge t_{0}$.
\end{proof}

\begin{remark} \label{ch4_remark_average}
From \textit{Lemma \ref{ch4_lem_average_undirected}}, it is clear that the group decision value $C$ satisfies $ \lim_{t \rightarrow \infty} x_{i} (t) = C = \frac{1}{N} \sum_{i=1}^{N} x_{i} (t_{0})$ $\forall i \in \mathcal{V}$ if the consensus is reached.
\end{remark}

As mentioned in the introduction section, the overall network consisting of $N$ agents (\ref{ch4_sys}) contains unachievable equilibrium points under output saturations, and thus, global consensus may not be realized.
Therefore, this paper investigates necessary and sufficient initial condition for the achievement of consensus, that is the exact domain of attraction.
Let $x = [ x_{1} ,..., x_{N}] \in \mathbf{R}^{N}$, and denote the state trajectory of agents (\ref{ch4_sys}) as $\phi (t, x)$ that starts at initial state $x$ at $t = t_{0}$.
Then, the domain of attraction of consensus, denoted by $\mathcal{X}$, is defined as the set of all points $x$ such that $\phi ( t, x )$ is defined for all $t \ge t_{0}$ and $\lim_{t \rightarrow \infty} \phi (t,x) = C \mathbf{1}$.

In this paper, we consider the following assumption to avoid the trivial solution.
\begin{assumption} \label{ch4_assumption}
Without loss of generality, we assume that the agents do not reach the consensus at $t = t_{0}$.
\end{assumption}

 \begin{remark} \label{rem_nonlinearity}
 Although this paper considers the saturation nonlinearity, the analysis of this paper can be easily applied to any bounded nonlinearities, which are strictly increasing within the bounds.
 \end{remark}
 

 

\section{Fixed Graph} \label{sec_fixed}


In this section, we deal with the consensus problem under the undirected and fixed graph.
Before we analyze the consensus, we consider the following lemma, which can be proved similar to the proof of \textit{Lemma 3.1} in \cite{Ren:2008}:
\begin{lemma} \label{lem_double_potential}
For an undirected graph and any $a_{i}, b_{i} \in \mathbf{R}$, $i = 1,...,N$, we have
\begin{align}
\sum_{i=1}^{N} \! \sum_{j=1}^{N} \alpha_{ij} ( a_{i} - a_{j} ) ( b_{i} - b_{j} ) \! = \! 2 \! \sum_{i=1}^{N} \! \sum_{j=1}^{N} \alpha_{ij} a_{i} ( b_{i} - b_{j} ).
\end{align}
\end{lemma}
\begin{proof}
Since the graph is undirected, $\alpha_{ij} = \alpha_{ji}$. Therefore we have
\begin{align}
\sum_{i=1}^{N} \sum_{j=1}^{N} \alpha_{ij} ( a_{i} - a_{j} ) ( b_{i} - b_{j})
&= \sum_{i=1}^{N} \sum_{j=1}^{N} \alpha_{ij} a_{i} ( b_{i} - b_{j} ) 
- \sum_{i=1}^{N} \sum_{j=1}^{N} \alpha_{ij} a_{j} ( b_{i} - b_{j} ) \nonumber\\
&= \sum_{i=1}^{N} \sum_{j=1}^{N} \alpha_{ij} a_{i} ( b_{i} - b_{j} ) 
+ \sum_{j=1}^{N} \sum_{j=1}^{N} \alpha_{ji} a_{j} ( b_{j} - b_{i} )  \nonumber\\
&= 2 \sum_{i=1}^{N} \sum_{j=1}^{N} \alpha_{ij} a_{i} ( b_{i} - b_{j} ),
\end{align}
which complete the proof.
\end{proof}

We next consider the following lemma that will be used to construct Lyapunov function:
\begin{lemma} \label{lem_integral_positive}
For any constants $a,b$ with $| a | \le s_{i}$,
\begin{align} \label{lem_int_positive}
\int_{a}^{b} ( \mbox{sat}_{i} ( \omega ) - a ) d \omega \ge 0,
\end{align}
and the equality holds only when $a = b$.
\end{lemma}
\begin{proof}
We consider the following two cases.\\
1) $a \le b \le s_{i}$.\\
In this case, (\ref{lem_int_positive}) can be rewritten as
\begin{align}
\int_{a}^{b} ( \omega - a ) d \omega =  \frac{1}{2} ( b^{2} - a^{2} - 2 a b + 2 a^{2} ) 
= \frac{1}{2} (b - a)^{2}
\end{align}
2) $a \le s_{i} < b$. \\
In this case, (\ref{lem_int_positive}) can be rewritten as
\begin{align}
\int_{a}^{s_{i}} \! ( \omega \! - \! a ) d \omega \! + \! \int_{s_{i}}^{b} ( s_{i} \! - \! a ) d \omega 
= \frac{1}{2}  (s_{i} \! - \! a )^{2} \! + \! ( s_{i} \! - \! a ) ( b \! - \! s_{i} )
 \ge \frac{1}{2} (s_{i} - a )^{2}
\end{align}
For $b \le a$, we can similarly derive (\ref{lem_int_positive}) holds, which complete the proof.
\end{proof}

Then, we are now ready to state the following result.

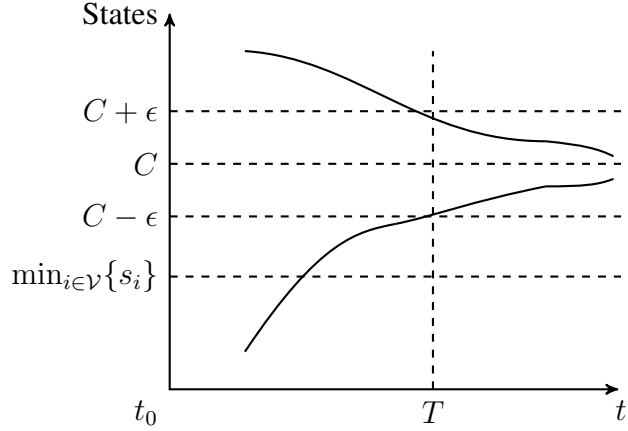
\begin{figure}[!tb]
\centering
\begin{tikzpicture} [node distance=1.5cm, thick,main node/.style={circle,fill=white!20,draw,font=\sffamily\scriptsize}]

 \coordinate (x) at (0,5);
 \coordinate (t) at (6,0);
 \draw[<->,>=stealth',shorten >=1pt,auto] (x) node[left] {States} -- (0,0) node[below left] {$t_{0}$} -- (t) node[below] {$t$};

 \draw[dashed] (0,1.5) node[left] {$\min_{i \in \mathcal{V}} \{s_{i} \}$} -- (6,1.5);
 \draw[dashed] (0,3) node[left] {$C$} -- (6,3);

 \draw (1,4.5) .. controls (2.5,4.4) and (3.3,3.3)  .. (5,3.3);
 \draw (5,3.3) .. controls (5.5,3.25) and (5.7,3.2) .. (5.9,3.1);

 \draw (1,0.5) .. controls (2,2) and (2.5,2.1) .. (3,2.2);
 \draw (3,2.2) .. controls (3.5,2.3) and (4,2.5) .. (5,2.7);
 \draw (5,2.7) .. controls (5.5, 2.7) and (5.7,2.72) .. (5.9,2.8);

 \draw[dashed] (3.5,0) node[below] {$ T$} -- (3.5,4.5);

 \draw[dashed] (0,2.3) node[left] {$C - \epsilon$} -- (6,2.3);
 \draw[dashed] (0,3.7) node[left] {$C + \epsilon$} -- (6,3.7);

\end{tikzpicture}
\caption{Illustration of the proof of \textit{Theorem~\ref{ch4_the_undirected_fixed}}} \label{necessity_graph}
\end{figure}

\begin{theorem} \label{ch4_the_undirected_fixed}
Suppose the graph $\mathcal{G}$ is undirected and connected. Then, the group of $N$ agents (\ref{ch4_sys}) achieves the consensus, 
if and only if 
\begin{align}
x (t_{0}) \in \mathcal{X} := \left\{ x (t_{0}) \in \mathbf{R}^{N} : \frac{1}{N} \left| \sum_{i=1}^{N} x_{i} (t_{0}) \right| \le \min_{i \in \mathcal{V}} \{ s_{i} \} \right\}.
\end{align}
\end{theorem}
\begin{proof}
(Necessity) We prove the necessity by a contradiction.
Assume that the agents achieve the consensus with the decision value $C = \frac{1}{N}  \sum_{i=1}^{N} x_{i} (t_{0}) > \min_{i \in \mathcal{V}} \{ s_{i} \}$, i.e., $\lim_{t \rightarrow \infty} x_{i} (t) = C, \forall i \in \mathcal{V}$.
Then, due to the continuity of $x_{i} (t)$, for any $\epsilon > 0$, there exists $T > 0$ such that $| x_{i} (t) - C | < \epsilon$ whenever $t \ge T$.
Since $\epsilon$ can be arbitrarily small number, we choose $\epsilon > 0$ such that $C - \epsilon \ge \min_{i \in \mathcal{V}} \{ s_{i} \}$ and the agents do not reach the consensus at $t = T$ (see Fig.~\ref{necessity_graph}). 
Then, the proof is divided into two cases.

1) Homogeneous agents.

Note that $x_{i} \ge s$ implies $y_{i} = \mbox{sat} ( x_{i} ) = s$.
Therefore, the agent $i$, $\forall i \in \mathcal{V}$, is given by for $t \ge T$
\begin{align}
\dot{x}_{i} =  \sum_{j=1}^{N} \alpha_{ij} ( y_{j} - y_{i} ) = \sum_{i=1}^{N} \alpha_{ij} ( s - s )
= 0,
\end{align}
which implies $x_{i} (t) = x_{i} ( T )$, $\forall t \ge  T$.
However, we have assumed that the agents do not reach the consensus, which is a contradiction.

2) Heterogeneous agents.

Let $i' = \argmin_{i \in \mathcal{V}} \{ s_{i} \}$.
Then, $x_{i} \ge  s_{i'}$, $\forall i \in \mathcal{V}$, implies $y_{i'} = \mbox{sat}_{i'} ( x_{i'} ) = s_{i'}$ and $y_{j} \ge s_{i'}$, $\forall j \in \mathcal{V} \setminus \{ i' \}$.
Therefore, the agent $i'$ is given by for $t \ge T$
\begin{align}
\dot{x}_{i'} = \sum_{j=1}^{N} \alpha_{i'j} ( y_{j} - y_{i'} ) = \sum_{j=1}^{N} \alpha_{i' j} ( y_{j} - s_{i'} )
\ge 0,
\end{align}
and $\dot{x}_{i'} = 0$ only when $y_{j} = s_{i'}$, $\forall j \in \mathcal{N}_{i'}$.
Therefore, since the graph is connected, the consensus is reached only when $x_{i} = s_{i'}$, $\forall i \in \mathcal{V}$, that is $\lim_{t \rightarrow \infty} x_{i} (t) = s_{i'} = \frac{1}{N} \sum_{i=1}^{N} x_{i} (t)$.
Since we assume that the decision value $C > s_{i'}$, this is a contradiction.

For $\frac{1}{N}  \sum_{i=1}^{N} x_{i} (t_{0}) < \max_{i \in \mathcal{V}} \{ - s_{i} \}$, 
we can derive the necessity following the same process as above, which completes the proof.

(Sufficiency) Let $x^{*} = \frac{1}{N} \sum_{i=1}^{N} x_{i} (t_{0})$, and $i' = \argmin_{i \in \mathcal{V}} \{ s_{i} \}$, and assume that $|x^{*} | \le s_{i'}$.
Consider the following Lypuanov function candidate:
\begin{align} \label{ch4_the_undirected_fixed_Lyapunov}
V =   2 \sum_{i=1}^{N} \int_{x^{*}}^{x_{i}} ( \mbox{sat}_{i} ( \omega ) - x^{*} ) d \omega.
\end{align}
From \textit{Lemma \ref{lem_integral_positive}} and the fact that $|x^{*}| \le s_{i}$, $\forall i \in \mathcal{V}$, we know that $V \ge 0$ and $V = 0$ only when $x_{i} = x^{*}$, $\forall i \in \mathcal{V}$.
We next consider the time derivative of $V$ as follows:
\begin{align}
\dot{V} =& 2 \sum_{i=1}^{N} \mbox{sat}_{i} ( x_{i} ) \dot{x}_{i} - 2 x^{*} \sum_{i=1}^{N} ( \dot{x}_{i} - \dot{x}^{*}) \nonumber\\
=&  2 \sum_{i=1}^{N}  \sum_{j=1}^{N} \alpha_{ij} y_{i} ( y_{j} - y_{i} ),
\end{align}
where we have used the fact that the average value is invariant, i.e.,  $\frac{1}{N} \sum_{i=1}^{N} \dot{x}_{i} = \dot{x}^{*} = 0$.
Then, from \textit{Lemma \ref{lem_double_potential}} with $a_{i} = b_{i} = y_{i}$, it follows that
\begin{align}
\dot{V} = - \sum_{i=1}^{N} \sum_{j=1}^{N} \alpha_{ij} ( y_{i} - y_{j} )^{2},
\end{align}
which implies $\dot{V} \le 0$.
Let $\mathcal{M} := \{ x \in \mathbf{R}^{N} : \dot{V} = 0 \}$.
Therefore, by LaSalle Invariance Principle, any solution of $x_{i} (t)$, $\forall i \in \mathcal{V}$, will converge to the largest invariant set inside $\mathcal{M}$.
We next prove that $x_{i} = x_{j}$, $\forall i,j \in \mathcal{V}$, is an unique equilibrium, which implies the consensus is reached.
Since the graph is connected, $\dot{V} \equiv 0$ is equivalent to $(y_{i} - y_{j} ) = ( \mbox{sat}_{i} (x_{i}) - \mbox{sat}_{j} ( x_{j} ) ) = 0$, $\forall i ,j \in \mathcal{V}$,
which, in turn, implies that $\mathbf{sat}  ( x ) \in \spand \{ \mathbf{1} \}$, where $\mathbf{sat} ( x)  = [ \mbox{sat}_{1} ( x_{1} ),..., \mbox{sat}_{N} ( x_{N} )]^{T}$.
Then, the rest of the proof is divided into two cases.

1) Homogeneous case.

$\mathbf{sat} ( x ) \in \spand \{ \mathbf{1} \}$ when
(a) $x \in \Omega_{u} := \Omega_{u^{+}} \cup \Omega_{u^{-}}$, where
$\Omega_{u^{+}} := \{ x \in \mathbf{R}^{N} : x \ge s \mathbf{1} , x \neq s \mathbf{1} \}$ and
$\Omega_{u^{-}} := \{ x \in \mathbf{R}^{N} : x \le - s \mathbf{1} , x \neq - s \mathbf{1} \}$,
(b) $x \in \Omega_{a} := \{ x \in \mathbf{R}^{N} : x = C \mathbf{1}, | C | \le s  \}$.
In case (a), the average value is given by $\frac{1}{N} \left| \sum_{i=1}^{N} x_{i} \right| > s$.
Since the average value is invariant, if $\frac{1}{N} \left| \sum_{i=1}^{N} x_{i} (t_{0}) \right| \le s$, then the case (a) can not be realized.
Therefore, $\mathbf{sat} (x) \in \spand \{ \mathbf{1} \}$, only when $x \in \Omega_{a}$, which implies $( x_{i} - x_{j} ) \equiv 0$, $\forall i , j \in \mathcal{V}$.


2) Heterogeneous case.

Since $|\mbox{sat}_{i'} ( x_{i'} (t) )| \le s_{i'}$, $\forall t \ge t_{0}$, where $i' = \argmin_{i \in \mathcal{V}} \{ s_{i} \}$,
the condition $( \mbox{sat}_{i} (x_{i}) - \mbox{sat}_{j} ( x_{j} )) \equiv 0$, $\forall i,j \in \mathcal{V}$, is equivalent to $( \mbox{sat}_{i'} ( x_{i'}) - \mbox{sat}_{j} ( x_{j} )) \equiv 0$, $\forall j \in \mathcal{V}$.
Therefore, $\mathbf{sat} ( x ) \in \spand \{ \mathbf{1} \}$ when
(a) $x \in \Omega_{i} := \Omega_{u^{+}} \cup \Omega_{u^{-}}$, where
$\Omega_{u^{+}} := \{ x \in \mathbf{R}^{N} : x_{i'} > s_{i'} \mbox{ and } x_{j} = s_{i'} , \forall j \in \mathcal{V} \setminus \{ i' \} \}$
and
$\Omega_{u^{-}} := \{ x \in \mathbf{R}^{N} : x_{i'} < - s_{i'} \mbox{ and } x_{j} = - s_{i'} , \forall j \in \mathcal{V} \setminus \{ i' \} \}$,
(b) $x \in \Omega_{a} := \{ x \in \mathbf{R}^{N} : x = C \mathbf{1} , | C | \le  s_{i'} \}$.
In case (a), the average value is given by $\frac{1}{N} \left| \sum_{i=1}^{N} x_{i} \right| > s_{i'}$, and thus, if $\frac{1}{N} \left| \sum_{i=1}^{N} x_{i} (t_{0}) \right| \le s_{i'}$, the cases (a) can not be realized.
Therefore, $\mathbf{sat} ( x) \in \spand \{ \mathbf{1} \}$, only when $x \in \Omega_{a}$, which implies $(x_{i} - x_{j} ) \equiv 0$, $\forall i , j \in \mathcal{V}$.


In summary, we have shown that $\dot{V} \le 0$ and $\dot{V} \equiv 0$ only when $(x_{i} - x_{j} ) \equiv 0$, $\forall i,j \in \mathcal{V}$.
Therefore, according to LaSalle Invariance Principle, we have $\lim_{t \rightarrow \infty} ( x_{i} (t) - x_{j} (t) ) = 0$, $\forall i,j \in \mathcal{V}$, which completes the proof.
\end{proof}


\section{Time-varying Graph} \label{sec_time-varying}


In this section, we consider a time-varying graph $\mathcal{G} (t) = ( \mathcal{V} , \mathcal{E} (t) , \mathcal{A} (t) )$ with the following assumption:
\begin{assumption} \label{ch4_assumption_time-varying}
For $\forall i,j \in \mathcal{V}$, $\alpha_{ij} (t)$ is a continuous function on $[0 , \infty)$ except for at most a set with measure zero.
\end{assumption}

Then, under \textit{Assumption \ref{ch4_assumption_time-varying}}, the set of discontinuity points for the right-hand side of (\ref{ch4_sys}) has measure zero.
Therefore, the Caratheodory solutions\footnote[1]{Caratheodory solutions are a generalization of classical solutions, and absolutely continuous functions of time.
Caratheodory solutions relax the classical requirement that the solution must follow the direction of the vector field at all times, see \cite{Filippov:1988,Cortes:2008} for details.} of (\ref{ch4_sys}) exist for arbitrary initial conditions, which satisfies for all $t \ge t_{0}$ the following integral equation  for $i \in \mathcal{V}$:
\begin{align}
x_{i} (t) = x_{i} (t_{0}) + \int_{t_{0}}^{t} \sum_{j=1}^{N} \alpha_{ij} ( \tau ) ( y_{j} (\tau ) - y_{i} ( \tau ) ) d \tau .
\end{align}


Before we analyze the consensus under the time-varying graph, we introduce some mathematical preliminaries.
Since the time-varying graph includes discontinuities to describe the switching phenomena,
the solution of $x_{i} (t)$ is not differentiable at the discontinuous points.
However, from \textit{Assumption \ref{ch4_assumption_time-varying}}, the upper Dini derivative of $x_{i}$ along the solution exists.
The upper Dini derivative of a function $f : (a,b) \rightarrow \mathbf{R}$ at $t$ is defined as
\begin{align}
D^{+} f (t) = \limsup_{\tau \rightarrow 0^+} \frac{f ( t + \tau ) - f (t)}{\tau}.
\end{align}

\begin{lemma} \cite{Rouche:1975}
Suppose $f (t)$ is continuous on $(a,b)$.
Then, $f(t)$ is nonincreasing on $(a,b)$ if and only if $D^{+} f(t) \le 0$, $\forall t \in ( a,b)$.
\end{lemma}

\subsection{Homogeneous Agents} \label{subsec_homo}

We first analyze the consensus of the agents (\ref{ch4_sys}) under the homogeneous condition, i.e., $s_{i} = s$, $\forall i \in \mathcal{V}$.
To solve this problem, we use the notations used in \cite{Cao:2011} as follows.

For any time $t$, let $M_{k} (t)$ be the $k$-th largest value of the components $x_{i} (t)$, that is, we rank $x_{i} (t)$ with descending order for each $t$ as follows:
\begin{align} \label{ch4_time-varying_homo_rank}
x_{i_{1}} (t) \ge x_{i_{2}} (t) \ge \cdots \ge x_{i_{N}} (t),
\end{align}
where $\{ i_{1} ,..., i_{N} \}$ is a permutation of $\{ 1,...,N \}$, and define
\begin{align} \label{ch4_time-varying_homo_M_k}
M_{k} (t) = x_{i_{k}} (t).
\end{align}
Note that the permutation $\{i_{k} : k \in \mathcal{V} \}$ depends on $t$, and the permutation $i_{k}$ is piecewise constant.
As a result, $M_{k} (t)$ is continuous for everywhere.
We further denote $S_{k} (t)$ as the sum of the first $k$ largest values of $x_{i} (t)$, i.e.,
\begin{align}
S_{0} (t) = 0 , \,~~ S_{k} (t) = \sum_{i=1}^{k} M_{i} (t) = M_{k} (t) + S_{k-1} (t).
\end{align}
Then, we first show the attractivity of equilibrium.
The proof follows from a similar argument in \cite{Cao:2011,Hendrickx:2013}.  

\begin{lemma} \label{ch4_lem_time-varying_ext_equilibria}
For the group of $N$ agents (\ref{ch4_sys}) under the homogeneous condition, there exists $x_{i}^{*}$ such that $\lim_{t \rightarrow \infty} x_{i} (t) = x_{i}^{*}$, and $x_{i}^{*} \in [ \min_{j \in \mathcal{V}} x_{j} (t_{0}) , \max_{j \in \mathcal{V}} x_{j} (t_{0}) ]$, $\forall i \in \mathcal{V}$.
\end{lemma}
\begin{proof}
Since $S_{i} (t)$ is absolutely continuous for almost everywhere, the derivative of $S_{m} (t)$ is given by
\begin{align} \label{eq_S_m}
D^{+} S_{m} (t) =& D^{+} \sum_{i=1}^{m} M_{i} (t) \nonumber\\
=& \sum_{k=1}^{m} \sum_{j=1}^{N} \alpha_{i_{k} j} (t) ( y_{j} (t) - y_{i_{k}} (t) ) \nonumber\\
=& \sum_{k=1}^{m} \sum_{j=1}^{m} \alpha_{i_{k} i_{j}} (t) ( y_{i_{j}} (t) - y_{i_{k}} (t) ) 
+ \sum_{k=1}^{m} \sum_{j=m+1}^{N} \alpha_{i_{k} i_{j}} (t ) ( y_{i_{j}} (t) - y_{i_{k}} (t) ) \nonumber\\
=& \sum_{k=1}^{m} \sum_{j=m+1}^{N} \alpha_{i_{k} i_{j}} (t) ( y_{i_{j}} (t) - y_{i_{k}} (t) ),
\end{align}
which implies $D^{+} S_{m} (t) \le 0$.
Therefore, $S_{m} (t)$ is nonincreasing function.
Moreover, $S_{m} (t)$ is bounded below, $S_{m} (t)$ converges as $t \rightarrow \infty$.
Since $M_{m} (t) = S_{m} (t) - S_{m-1}(t)$, then $M_{m} (t)$ converges, too.
This implies that every $M_{i} (t)$ converges to a limit $\lim_{t \rightarrow \infty} M_{i} (t) = M_{i}^{*}$.
Then, from the definition of $M_{i} (t)$, each $x_{i} (t)$ must converge to one of the values of $M_{j}^{*}$.
Moreover, we can easily see from (\ref{eq_S_m}) that $D^{+} S_{1} (t) = D^{+} M_{1} (t) \le 0$ and $D^{+} S_{N} (t) = D^{+} M_{N} (t) + D^{+} S_{N-1} (t) = 0$, which implies $D^{+} M_{N} (t) \ge 0$ since $D^{+} S_{N-1} (t) \le 0$.
Note that $M_{1} (t) = \max_{i \in \mathcal{V}} x_{i} (t)$ and $M_{N} (t) = \min_{i \in \mathcal{V}} x_{i} (t)$.
Then, $D^{+} M_{1} (t) \le 0$ and $D^{+} M_{N} (t) \ge 0$ imply that $x_{i} (t) \in [ \min_{j\in\mathcal{V}} x_{j} (t_{0}), \max_{j \in \mathcal{V}} x_{j} (t_{0})]$, $\forall i \in \mathcal{V}, t \ge t_{0}$.
Therefore, we have $M_{i}^{*} \in  [ M_{N} (t_{0}) , M_{1} (t_{0})] = [ \min_{j\in\mathcal{V}} x_{j} (t_{0}), \max_{j \in \mathcal{V}} x_{j} (t_{0})]$.
\end{proof}

We next recall the integral graph $\bar{\mathcal{G}}_{[0, \infty)}$ in \textit{Definition \ref{def_integral_graph}}.
Let $\bar{L}$ be the Laplacian of $\bar{\mathcal{G}}_{[0, \infty)}$, and $y_{i}^{*} = \mbox{sat} ( x_{i}^{*} ) = \lim_{t \rightarrow \infty} \mbox{sat} ( x_{i} (t) )$.
Then, similarly to \textit{Lemma 4.3} in \cite{Cao:2011}, we have
\begin{lemma} \label{ch4_lem_time-varying_equilibria}
$x^{*} \in \{ x \in \mathbf{R}^{N} : \mathbf{sat} (x) \in \Ker \bar{L} \}$.
\end{lemma}
\begin{proof}
The proof follows from \textit{Lemma 4.3} in \cite{Cao:2011}.
Consider any two components $y_{i}^{*}$ and $y_{j}^{*}$.
Firstly, we assume that $x_{i}^{*} \neq x_{j}^{*}$.
Define $d: = | y_{i}^{*} - y_{j}^{*}|$.
Since $y^{*} = \lim_{t \rightarrow \infty} y(t)$, there is $T \ge 0$ for any $\epsilon > 0$ such that $|y_{i} - y_{i}^{*}| < \epsilon$ and $|y_{j} - y_{j}^{*}| < \epsilon$, for $t \ge T$. 
Then, $|y_{j} - y_{i} | \ge d - 2 \epsilon$, and thus
\begin{align}
\int_{T}^{\infty} \alpha_{ij} (\tau) ( y_{j} (\tau) - y_{i} ( \tau ) )^{2} d \tau \ge \int_{T}^{\infty} \alpha_{ij} ( \tau ) ( d - 2 \epsilon )^{2} d \tau.
\end{align}
Then, applying \textit{Lemma 3.6} in \cite{Cao:2011}, we have
\begin{align}
\int_{T}^{\infty} \alpha_{ij} ( \tau ) d \tau \le \frac{|| x(0) ||^{2}}{(d - 2 \epsilon)^{2}},
\end{align}
and also $\int_{0}^{\infty}\alpha_{ij} ( \tau ) d \tau < \infty$.
By the definition of $\bar{\mathcal{G}}_{[0, \infty)}$, it implies $\bar{\alpha}_{ij} = 0$ when $y_{i}^{*} = \mbox{sat} ( x_{i}^{*}) \neq y_{j}^{*} = \mbox{sat} (x_{j}^{*} )$. \\
Secondly, if $y_{i}^{*} = y_{j}^{*}$, then $y_{i}^{*} - y_{j}^{*} = 0$.
Therefore, we have
\begin{align}
\bar{\alpha}_{ij} ( y_{i}^{*} - y_{j}^{*} ) = 0,
\end{align}
which implies $\bar{L} y^{*} = \bar{L} \mathbf{sat} (x^{*} ) = 0$.
\end{proof}

Note that \textit{Lemma \ref{ch4_lem_time-varying_equilibria}} does not imply that $x^{*}$ is such that $\mathbf{sat} ( x^{*} ) = x^{*}$ and $\bar{L} x^{*} = 0$ due to the existence of unachievable equilibrium.
For example, for $x^{*} \ge s$ and $x^{*} \neq s \mathbf{1}$, $\bar{L} \mathbf{sat} ( x^{*} ) = \bar{L} s \mathbf{1} = 0$.
However, the following theorem shows that, when the consensus is reached, $\mathbf{sat} ( x^{*} ) = x^{*} = C \mathbf{1}$.

\begin{theorem} \label{ch4_the_time-varying_homo}
Suppose the undirected, time-varying graph $\mathcal{G}(t)$ is integrally connected over $[0, \infty)$, i.e.,  $\bar{\mathcal{G}}_{[0,\infty)}$ is connected. Then, the group of $N$ agents (\ref{ch4_sys}) under the homogeneous condition achieves the consensus, 
if and only if 
\begin{align}
x (t_{0}) \in \mathcal{X} := \left\{ x (t_{0}) \in \mathbf{R}^{N} : \frac{1}{N} \left| \sum_{i=1}^{N} x_{i} (t_{0}) \right| \le  s  \right\}.
\end{align}
\end{theorem}
\begin{proof}
From \textit{Lemma \ref{ch4_lem_time-varying_ext_equilibria}} and \textit{Lemma \ref{ch4_lem_time-varying_equilibria}}, we know that $\lim_{t \rightarrow \infty} x(t) = x^{*}$, where $\mathbf{sat} ( x^{*} ) \in \Ker \bar{L}$.
For the integrally connected graph, $\Ker \bar{L} = \spand \{ \mathbf{1} \}$.
Then, it follows that $\mathbf{sat} (x^{*}) \in \spand \{ \mathbf{1} \}$, which, in turn, implies $\Omega : = \Omega_{u} \cup \Omega_{a}$, where $\Omega_{u}$ and $\Omega_{a}$ are defined in the proof of \textit{Theorem \ref{ch4_the_undirected_fixed}}.
Then, since the average value is invariant, the necessity and the sufficiency directly follow from the case of the fixed graph.
%
\end{proof}

\subsection{Heterogeneous Agents} \label{subsec_hetero}

In this subsection, we consider the heterogeneous condition, that is a general case of the homogeneous condition. 
In this case,  we need the following additional assumption on the graph:
\begin{assumption} \label{ch4_assumption_hetero}
For any pair $(i,j) \in \mathcal{E}$,  $\alpha_{ij} (t) \in 0 \bigcup [ \alpha_{\min}, \alpha_{\max} ]$.
\end{assumption}
Moreover, we assume that without loss of generality, the agents are already sorted such that $s_{1} > \cdots > s_{N}$.
Otherwise, by rearranging the order of the agents, we have this form.

Let $\mathcal{V}_{k}$ be a subset of the node set $\mathcal{V}$ defined by $\mathcal{V}_{k} := \{ 1, 2,..., k \}$.
Then, we first consider the following lemma whose proof is given in Appendix.
\begin{lemma} \label{ch4_lem_time-varying_hetero_conv_N-1}
Suppose the graph $\mathcal{G} (t)$ is integrally connected with \textit{Assumption \ref{ch4_assumption_hetero}}.
Then, with the consensus algorithm (\ref{ch4_sys}), for any $x_{i} (t_{0}) \in \mathbf{R}$, $\forall i \in \mathcal{V}$, there exists a number $T \ge t_{0}$ such that  there holds $x_{i} \in ( - s_{i} , s_{i} )$, $\forall i \in \mathcal{V}_{N-1}$, $\forall t \ge T$.
Moreover, we have $\lim_{t \rightarrow \infty} | x_{i} (t) | \le s_{N}$, $\forall i \in \mathcal{V}_{N-1}$.
\end{lemma}


We next recall $M_{k} (t)$ and $S_{k} (t)$ defined in Section~\ref{subsec_homo}.
Then, similar to \textit{Lemma \ref{ch4_lem_time-varying_ext_equilibria}}, we can show the existence of limits as follows:
\begin{lemma} \label{ch4_lem_time-varying_ext_hete_limit}
For the group of $N$ agents (\ref{ch4_sys}) under the heterogeneous condition, there exists $x_{i}^{*}$ such that $\lim_{t \rightarrow \infty} x_{i} (t) = x_{i}^{*}$ $\forall i \in \mathcal{V}$.
Moreover, $x_{i}^{*} \in [ - s_{N}, s_{N} ]$, $\forall i \in \mathcal{V}_{N-1}$.
\end{lemma}
\begin{proof}
Since $S_{i} (t)$ is absolutely continuous for almost everywhere, following the proof of \textit{Lemma \ref{ch4_lem_time-varying_ext_equilibria}}, we can obtain
\begin{align} \label{ch4_lem_time-varying_ext_hete_limit_eq1}
D^{+} S_{m} (t) =& \sum_{k=1}^{m} \sum_{j=m+1}^{N} \alpha_{i_{k} i_{j}} (t) ( y_{i_{j}} (t) - y_{i_{k}} (t) ) .
\end{align}
Note that, in the heterogeneous case, (\ref{ch4_lem_time-varying_ext_hete_limit_eq1}) does not imply $D^{+} S_{m} (t) \le 0$ since $x_{i} \ge x_{j}$ does not imply $y_{i} = \mbox{sat}_{i} ( x_{i} ) \ge y_{j} = \mbox{sat}_{j} ( x_{j} )$,
e.g., for $x_{i} > x_{j} > s_{j} > s_{i}$, $y_{j} = s_{j} > y_{i} = s_{i}$.
However, from \textit{Lemma \ref{ch4_lem_time-varying_hetero_conv_N-1}}, we know that there exists $T \ge t_{0}$ such that $x_{i} (t) \in ( - s_{N-1} , s_{N-1} )$, $\forall i \in \mathcal{V}_{N-1}$, $\forall t \ge T$.
Therefore, for $t \ge T$, if $x_{i} \ge x_{j}$, then $y_{i} \ge y_{j}$, which implies $D^{+} S_{m} (t) \le 0$, $\forall t \ge T$.
Moreover, the average value is invariant, $x_{N} (t)$ is bounded, and $S_{m} (t)$ is bounded below.
Then, following the proof of \textit{Lemma \ref{ch4_lem_time-varying_ext_equilibria}}, we know that
each $x_{i} (t)$ must converge to one of the values of $M_{j}^{*}$.
Moreover, from \textit{Lemma \ref{ch4_lem_time-varying_hetero_conv_N-1}}, $\lim_{t \rightarrow \infty} | x_{i} (t) | \le s_{N}$ for all $i \in \mathcal{V}_{N-1}$. Therefore, the equilibrium of $x_{i}$ for all $i \in \mathcal{V}_{N-1}$ must be within the interval $[- s_{N} ,s_{N} ]$.
\end{proof}

\begin{lemma} \label{ch4_lem_time-varying_equilibria_heterogeneous}
$x^{*} \in \{ x \in \mathbf{R}^{N} : \mbox{sat}_{N} (x) \in \Ker \bar{L}, \,~~ |x_{i} | \le s_{N} \forall i \in \mathcal{V}_{N-1} \}$.
\end{lemma}
\begin{proof}
From \textit{Lemma \ref{ch4_lem_time-varying_equilibria}}, we know that $y^{*} \in \Ker \bar{L}$.
Since $\lim_{t \rightarrow \infty} | x_{i} (t) | = | x_{i}^{*} | \le s_{N}$, $\forall i \in \mathcal{V}_{N-1}$ from \textit{Lemma \ref{ch4_lem_time-varying_ext_hete_limit}}, $y_{i}^{*} = x_{i}^{*} = \mbox{sat}_{N} ( x_{i}^{*} )$, $\forall i \in \mathcal{V}_{N-1}$, which completes the proof.
\end{proof}

Then, we are now ready to state the following result.
\begin{theorem} \label{ch4_the_time-varying_hetero}
Suppose the undirected, time-varying graph $\mathcal{G}(t)$ is integrally connected over $[0, \infty)$ with \textit{Assumption \ref{ch4_assumption_time-varying}}.
Then, the group of $N$ agents (\ref{ch4_sys}) under the heterogeneous condition achieves the consensus, 
if and only if 
\begin{align}
x (t_{0}) \in \mathcal{X} := \left\{ x (t_{0}) \in \mathbf{R}^{N} : \frac{1}{N} \left| \sum_{i=1}^{N} x_{i} (t_{0}) \right| \le \min_{i \in \mathcal{V}} \{ s_{i} \} \right\}.
\end{align}
\end{theorem}
\begin{proof}
The necessity directly follows from the case of fixed graph.
Therefore, we will prove the sufficiency only.
From \textit{Lemma \ref{ch4_lem_time-varying_ext_hete_limit}}, \textit{Lemma \ref{ch4_lem_time-varying_equilibria_heterogeneous}} and the fact that $\Ker \bar{L} = \spand \{ \mathbf{1} \}$, 
there exists a constant $C$ such that $\lim_{t \rightarrow \infty} x_{i} = x_{i}^{*} = C \in [ - s_{N} ,s_{N} ]$ $\forall i \in \mathcal{V}_{N-1}$, but need not $\lim_{t \rightarrow \infty} x_{N} (t) = x_{N}^{*} = C$.
{Moreover, since the average value is invariant, we have $\lim_{t \rightarrow \infty} \frac{1}{N} \sum_{i=1}^{N} x_{i} (t ) = \frac{1}{N} \sum_{i=1}^{N} x_{i} (t_{0}) = \frac{1}{N} ( ( N-1 ) C + x_{N}^{*} )$.
Therefore, if $ \frac{1}{N} | \sum_{i=1}^{N} x_{i} (t_{0}) | \le \min_{i \in \mathcal{V}} \{ s_{i} \}$, then $x_{N}^{*} = C = \frac{1}{N} \sum_{i=1}^{N} x_{i} (t_{0})$ is unique equilibrium point, which completes the proof.}
\end{proof}

\subsection{Unachievable Equilibrium} \label{subsec_unach}
In this subsection, we investigate some properties of unachievable equilibrium for consensus.
For simplicity, we assume that $\frac{1}{N} \sum_{i=1}^{N} x_{i} (t_{0}) > \min_{i \in \mathcal{V}} \{ s_{i} \}$.

For the homogeneous case, as mentioned in the introduction section, the set of unachievable equilibrium is defined by $\Omega_{u^{+}} := \{ x \in \mathbf{R}^{N} : x \ge s, x \neq s \}$.
Then, it is clear from Section~\ref{subsec_homo} that $\lim_{t \rightarrow \infty} x_{i} (t) = x_{i}^{*} \in [ s , \max_{i \in \mathcal{V}} x_{i} (t_{0} )]$.
Moreover, the derivative of $|x_{i} (t)|$ is given by
\begin{align}
D^{+} | x_{i} (t) | \le& \sum_{j=1}^{N} \alpha_{ij} (t) ( |y_{j} (t)| - | y_{i} (t) | ) \nonumber\\
\le&  \sum_{j=1}^{N} \alpha_{ij} (t) ( s - | y_{i} (t) | ).
\end{align}
{Then, for $|y_{i} (t) | = s$,  $D^{+} | x_{i} (t) | \le 0$,
which implies that the set $\mathcal{O}_{i} := \{ x_{i} : | x_{i} | \le s \}$ is a positively invariant set, i.e.,  if $x_{i} (t^{*}) \in \mathcal{O}_{i}$, then $x_{i} (t ) \in \mathcal{O}_{i}$ $\forall t \ge t^{*}$.}
Therefore, $\lim_{t \rightarrow \infty} x_{i} (t) = s$,  $\forall i \in \{ i \in \mathcal{V}: x_{i} (t_{0}) \le s \}$, and the remaining agents converge to the interval $[ s, \max_{i \in \mathcal{V}} x_{i} (t_{0})]$.

{For the heterogeneous case, according to Section~\ref{subsec_hetero}, the set of unachievable equilibrium is defined by $\Omega_{u^{+}} := \{ x \in \mathbf{R}^{N} : x_{i} = s_{N}, \forall i \in \mathcal{V}_{N-1}, \, x_{N} > s_{N} \}$, which implies, for any $x_{i} (t_{0})$, $\lim_{t \rightarrow \infty} x_{i} (t) = s_{N}$, $\forall i \in \mathcal{V}_{N-1}$, and $\lim_{t \rightarrow \infty} x_{N} (t) > s_{N}$.
Moreover, the invariance of the average value implies $\lim_{t \rightarrow \infty} \sum_{i=1}^{N} x_{i} (t) = \lim_{t \rightarrow \infty} x_{N} (t) + (N-1) s_{N} = \sum_{i=1}^{N} x_{i} (t_0)$, which gives $\lim_{t \rightarrow \infty} x_{N} (t) = \sum_{i=1}^{N} x_{i} (t_{0}) - ( N-1 ) s_{N}$.}

\section{Extensions} \label{extension}

\subsection{Double-integrator agents}
Since many real systems are controlled by the acceleration rather than the velocity, this subsection extends the previous results to the double-integrator modeled agents.
Consider the following group of $N$ double-integrator modeled agents:
\begin{align} \label{sys_double}
\dot{x}_{i} =& v_{i} \nonumber\\
\dot{v}_{i} =& u_{i}, \,~~~ i \in \mathcal{V}:= \{ 1,...,N \},
\end{align}
where $x_{i}, v_{i}, u_{i} \in \mathbf{R}$ are the position (or angle), velocity (or angular velocity), and control input of the agent $i$, respectively.
It was shown that the following consensus algorithm proposed in \cite{Ren_int:2007}
\begin{align}
u_{i} = \sum_{j=1}^{N} \alpha_{ij} \left( ( x_{j} - x_{i} ) + ( v_{j} - v_{i} ) \right),
\end{align}
solves the consensus problem for any $x_{i} (t_{0})$ and $v_{i} (t_{0})$, specifically,
$ x_{i} (t) \rightarrow  \frac{1}{N} \sum_{i=1}^{N} x_{i} ( t_{0} ) + t \frac{1}{N} \sum_{i=1}^{N} v_{i} (t_{0})$ and $ v_{i} (t) \rightarrow \frac{1}{N} \sum_{i=1}^{N} v_{i} (t_{0})$, $\forall i \in \mathcal{V}$.
However, in the presence of the measurement saturations, the consensus may not be reached due to the existence of unachievable equilibrium.
In this subsection, we assume that the measurements of the velocities have the homogeneous saturation levels and thus consider the following consensus algorithm:
\begin{align} \label{consensus_double}
u_{i} =& \sum_{j=1}^{N} \alpha_{ij} \left( ( x_{j} - x_{i} ) + ( y_{j} - y_{i} ) \right), \nonumber\\
y_{i} =& \mbox{sat} ( v_{i} ).
\end{align}



Then, by extending \textit{Theorem \ref{ch4_the_undirected_fixed}}, we have the following result:

\begin{theorem}
Suppose the graph is undirected and connected. Then, the group of $N$ agents (\ref{sys_double}) under the consensus algorithm (\ref{consensus_double}) achieves the consensus, i.e., $\lim_{t \rightarrow \infty} ( x_{i} - x_{j}) = 0$ and $\lim_{t \rightarrow \infty} ( v_{i} - v_{j} ) = 0$, $\forall i,j \in \mathcal{V}$,
if and only if 
\begin{align}
(x (t_{0}), v(t_{0})) \in \mathcal{X} := \left\{ (x (t_{0}), v (t_{0}) ) \in \mathbf{R}^{2N} : \frac{1}{N} \left| \sum_{i=1}^{N} v_{i} (t_{0}) \right| \le s \right\}.
\end{align}
\end{theorem}
\begin{proof}
Since the average of all velocities is invariant, the necessity directly follows from \textit{Theorem \ref{ch4_the_undirected_fixed}}.
Therefore, we will prove the sufficiency only.

Consider the following Lyapunov function candidate:
\begin{align}
V = \frac{1}{2} \sum_{i=1}^{N} \sum_{j=1}^{N} \alpha_{ij} (x_{i} - x_{j})^{2} +  \sum_{i=1}^{N} v_{i}^{2}.
\end{align}
Then, the time derivative of $V$ is given by
\begin{align}
\dot{V} =&  \sum_{i=1}^{N} \sum_{j=1}^{N} \alpha_{ij} (x_{i} - x_{j}) ( \dot{x}_{i} - \dot{x}_{j}) + 2 \sum_{i=1}^{N} v_{i} \dot{v}_{i} \nonumber\\
=&  \sum_{i=1}^{N} \sum_{j=1}^{N} \alpha_{ij} (x_{i} - x_{j}) ( v_{i} - v_{j} ) 
+ 2  \sum_{i=1}^{N} \sum_{j=1}^{N} \alpha_{ij} v_{i} ( x_{j} - x_{i} + y_{j} - y_{i} ) .
\end{align}
Note that, by applying \textit{Lemma \ref{lem_double_potential}} with $a_{i} = v_{i}$ and $b_{i} = x_{i} + y_{i}$, we have
\begin{align}
2  \sum_{i=1}^{N} \sum_{j=1}^{N} \alpha_{ij} v_{i} ( x_{j} - x_{i} + y_{j} - y_{i} )
= - \sum_{i=1}^{N} \sum_{j=1}^{N} \alpha_{ij} ( v_{i} - v_{j} ) ( x_{i} - x_{j} + y_{i} - y_{j} ),
\end{align}
which gives
\begin{align}
\dot{V} =&  - \sum_{i=1}^{N} \sum_{j=1}^{N} \alpha_{ij} ( v_{i} - v_{j} ) ( y_{i} - y_{j} ).
\end{align}
Since the saturation function satisfies the incremental passive condition \cite{Pavlov:2008}, i.e.,
\begin{align}
(v_{i} - v_{j}) ( \mbox{sat} (v_{i}) - \mbox{sat} ( v_{j} ) ) \ge 0, \mbox{ for any } i,j \in \mathcal{V},
\end{align}
we have $\dot{V} \le 0$.
Let $\mathcal{M} := \{ (x,v) \in \mathbf{R}^{2N} : \dot{V} = 0 \}$.
Then, $\dot{V} \equiv 0$ implies that either $( v_{i} - v_{j} ) \equiv 0$ or $(\mbox{sat} (v_{i}) - \mbox{sat} (v_{j} ) ) \equiv 0$, $\forall i,j \in \mathcal{V}$.
Since the average of all velocities is invariant, we can prove from the proof of \textit{Theorem \ref{ch4_the_undirected_fixed}} that if $\frac{1}{N} \left| \sum_{i=1}^{N} v_{i} (t_{0}) \right| \le s$, then  $(\mbox{sat} (v_{i}) - \mbox{sat} (v_{j} ) ) \equiv 0$, $\forall i,j \in \mathcal{V}$, only when $( v_{i} - v_{j} ) \equiv 0$, $\forall i,j \in \mathcal{V}$.
Moreover, $( v_{i} - v_{j} ) \equiv 0$, $\forall i,j \in \mathcal{V}$, implies $( \dot{v}_{i} - \dot{v}_{j} ) \equiv 0$, $\forall i,j \in \mathcal{V}$, {in an invariant set within $\mathcal{M}$},
which gives that $\dot{v} \in \spand \{ \mathbf{1} \}$.
Note that the average of all velocities is invariant, i.e., $ \mathbf{1}^{T} \dot{v} = 0$, and thus, $\dot{v}$ is orthogonal to $\mathbf{1}$.
Therefore, we can conclude that $\dot{v} \equiv 0$, and thus, from (\ref{consensus_double}) and the fact that $( v_{i} - v_{j} ) \equiv 0$, it follows that $\dot{v}_{i} \equiv - \sum_{j=1}^{N} \alpha_{ij} x_{ij} \equiv 0$.
As a result, we have $\sum_{i=1}^{N} x_{i} \sum_{j=1}^{N} \alpha_{ij} x_{ij} \equiv 0$, which implies from \textit{Lemma \ref{lem_double_potential}} that $\frac{1}{2} \sum_{i=1}^{N} \sum_{j=1}^{N} \alpha_{ij} ( x_{i} - x_{j} )^{2} \equiv 0$.
Since the graph is connected, we can conclude that $( x_{i} - x_{j} ) \equiv 0$, $\forall i,j \in \mathcal{V}$.

In summary, we have shown that $\dot{V} \le 0$ and $\dot{V} \equiv 0$ only when $(x_{i} - x_{j} ) \equiv 0$ and $( v_{i} - v_{j} ) \equiv 0$, $\forall i,j \in \mathcal{V}$.
Therefore, according to Lasalle Invariance Principle, we have $ \lim_{t \rightarrow \infty} (x_{i} (t) - x_{j} (t) ) = 0$ and $\lim_{t \rightarrow \infty} ( v_{i} (t) - v_{j} (t) ) = 0$, $\forall i,j \in \mathcal{V}$, which completes the proof.
\end{proof}

\subsection{Directed graph} \label{subsec_directed}

In this subsection, we consider the single-integrator modeled agents as in (\ref{ch4_sys}) with a directed graph.
Let $p = [ p_{1} ,...,p_{N}]^{T}$ be the left eigenvector of its Laplacian matrix $L$ associated with eigenvalue $\lambda_{1} = 0$, and $\sum_{i=1}^{N} p_{i} =1$.
Note that $p$ is positive \cite{Lewis:2014}, and it is clear that the weighted average of all agents' states defined by $\sum_{i=1}^{N} p_{i} x_{i} (t)$ is invariant.
Then, we have the following lemma, which can be proved from the proof of \textit{Lemma 7.7} in \cite{Lewis:2014}:
\begin{lemma}  \label{lem_directed_potential}
For a strongly connected, directed graph, and any $y_{i} \in \mathbf{R}$, $i = 1,...,N$, we have
\begin{align}
2 \sum_{i=1}^{N} \sum_{j=1}^{N} p_{i} \alpha_{ij} y_{i} ( y_{i} - y_{j} ) = \sum_{i = 1}^{N} \sum_{j=1}^{N} p_{i} \alpha_{ij} ( y_{i} - y_{j} )^{2}.
\end{align}
\end{lemma}

\begin{theorem} \label{the_directed}
Suppose the graph is directed and strongly connected.
Then, the group of $N$ agents (\ref{ch4_sys}) achieves the consensus,
if and only if 
\begin{align}
x (t_{0}) \in \mathcal{X} := \left\{ x (t_{0}) \in \mathbf{R}^{N} : \left| \sum_{i=1}^{N} p_{i} x_{i} (t_{0}) \right| \le \min_{i \in \mathcal{V}} \{ s_{i} \} \right\}.
\end{align}
\end{theorem}
\begin{proof}
Since the weighted average of all agents' states is invariant, the necessity can be proved similar to the case of the undirected graph.
Therefore, we will prove the sufficiency only.

Let $x^{*} = \sum_{i=1}^{N} p_{i} x_{i} (t_{0})$, and assume that $|x^{*} | \le \min_{i \in \mathcal{V}} s_{i}$.
Consider the following Lypuanov function candidate:
\begin{align} \label{direct_Lyapunov}
V = 2 \sum_{i=1}^{N} p_{i} \int_{x^{*}}^{x_{i}} ( \mbox{sat}_{i} ( \omega ) - x^{*} ) d  \omega.
\end{align}
Since $p_{i} > 0$, $i = 1,...,N$, from \textit{Lemma \ref{lem_integral_positive}}, we know that $V \ge 0$.
Note that the weighted average is invariant, i.e., $\frac{1}{N} \sum_{i=1}^{N} p_{i} \dot{x}_{i} = \dot{x}^{*} = 0$.
Therefore, the time derivative of $V$ given by
\begin{align}
\dot{V} =& 2 \sum_{i=1}^{N} p_{i} \mbox{sat}_{i} ( x_{i} ) \dot{x}_{i} - 2 x^{*} \sum_{i=1}^{N} p_{i} ( \dot{x}_{i} - \dot{x}^{*}) \nonumber\\
=& 2 \sum_{i=1}^{N}  \sum_{j=1}^{N} \alpha_{ij} p_{i}  y_{i} ( y_{j} - y_{i} ).
\end{align}
Then, from \textit{Lemma \ref{lem_directed_potential}}, it follows that
\begin{align}
\dot{V} = - \sum_{i=1}^{N} \sum_{j=1}^{N} p_{i} \alpha_{ij} ( y_{i} - y_{j} )^{2},
\end{align}
which implies $\dot{V} \le 0$.
Let $\mathcal{M} := \{ x \in \mathbf{R}^{N} : \dot{V} = 0 \}$.
Then, since the graph is strongly connected, $\dot{V} \equiv 0$ implies that $(y_{i} - y_{j}) \equiv ( \mbox{sat}_{i} ( x_{i} ) - \mbox{sat}_{j} ( x_{j} ) ) \equiv 0$, $\forall i,j \in \mathcal{V}$.
Then, similar to the proof of \textit{Theorem \ref{ch4_the_undirected_fixed}}, we can prove that $\dot{V} \equiv 0$ only when $(x_{i} - x_{j}) \equiv 0$, $\forall i,j \in \mathcal{V}$.
Therefore, applying LaSalle Invariance Principle gives $\lim_{t \rightarrow \infty} ( x_{i} (t) - x_{j} (t) ) = 0$, $\forall i,j \in \mathcal{V}$, which completes the proof.
\end{proof}

\section{Simulation results} \label{simulation}

\subsection{Fixed Graph}

We consider a group of $50$ agents whose topology is fixed, undirected and connected, whose second smallest and largest eigenvalues are given by $\lambda_{2} =     0.5327$ and $\lambda_{50} =12.3631$, respectively.

We first consider the homogeneous agents with $s = 1$.
The initial conditions are uniformly distributed on the interval $[-10 , 10]$.
Fig.~\ref{Fig:fixed_homo} shows the simulation results with the average values are (a) $-0.9821$ and (b) $1.3060$.
Then, the case (a) satisfies the condition in \textit{Theorem \ref{ch4_the_undirected_fixed}}, and thus the agents achieve the consensus.
However, the case (b) does not satisfy the condition in \textit{Theorem \ref{ch4_the_undirected_fixed}}, and consequently, the consensus is not reached.

We next consider the heterogeneous agents.
We choose the saturation levels on the interval $s_{i} \in [ 1, 7]$, $\forall i \in \mathcal{V}$ and $\min_{i \in \mathcal{V}} \{s_{i}\} = 1$.
With the same initial conditions as used in the homogeneous case, 
the simulation result is given in Fig.~\ref{Fig:fixed_hetero}.
From \textit{Theorem \ref{ch4_the_undirected_fixed}}, it is clear that the case (a) achieves the consensus, but the case (b) is not.
Moreover, in the case (b), there are $3$ agents whose saturation levels are $1$.
As we discussed in Section~\ref{subsec_unach}, the agents except for $3$ agents, whose saturation levels are $1$, converge to $1$.

\begin{figure}[!tb]
\begin{center}
\subfloat[$\frac{1}{N} \sum_{i=1}^{N} x_{i} (0) = -0.9821$]{\includegraphics[width=0.75\textwidth]{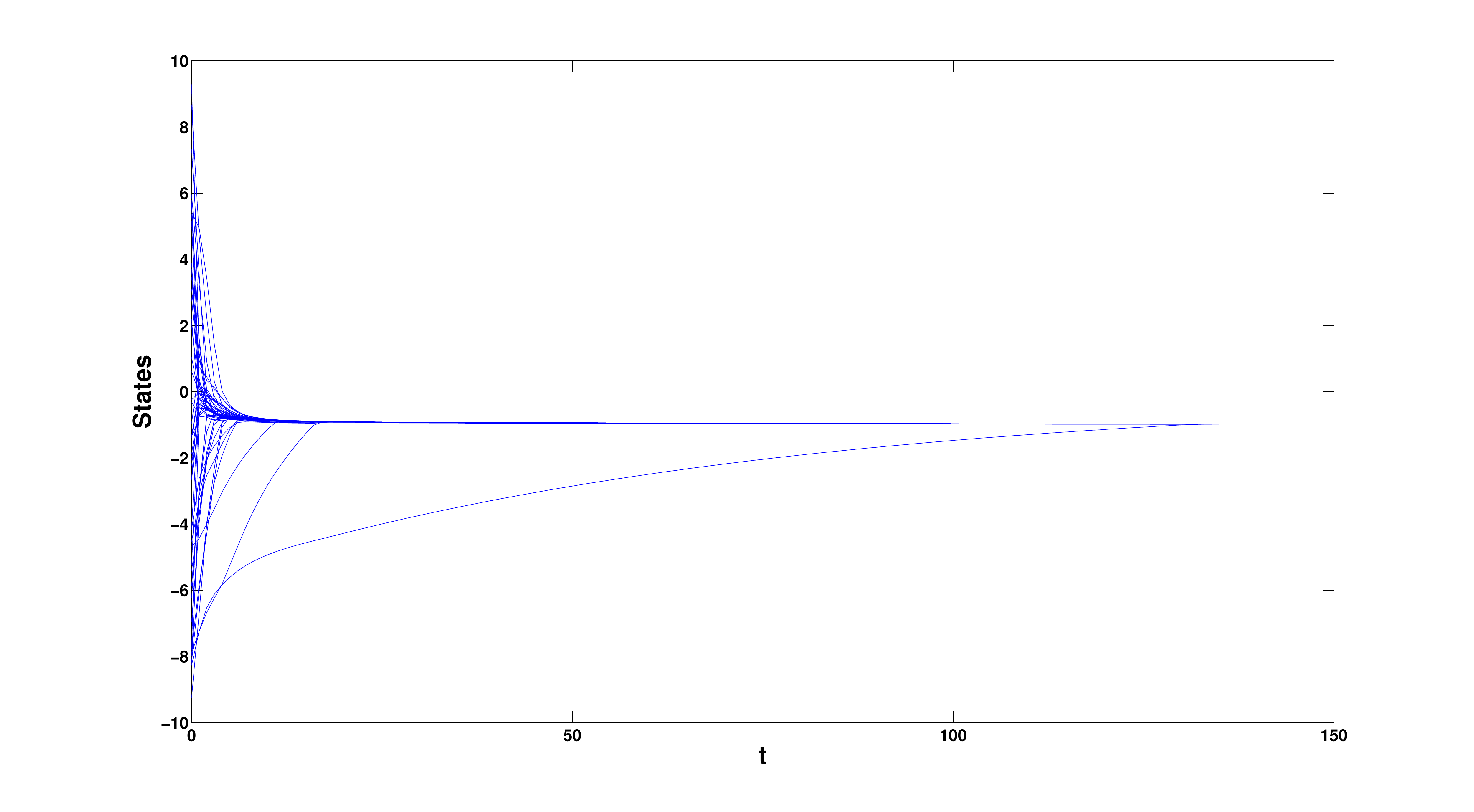}}\\
\subfloat[$\frac{1}{N} \sum_{i=1}^{N} x_{i} (0) = 1.3060$]{\includegraphics[width=0.75\textwidth]{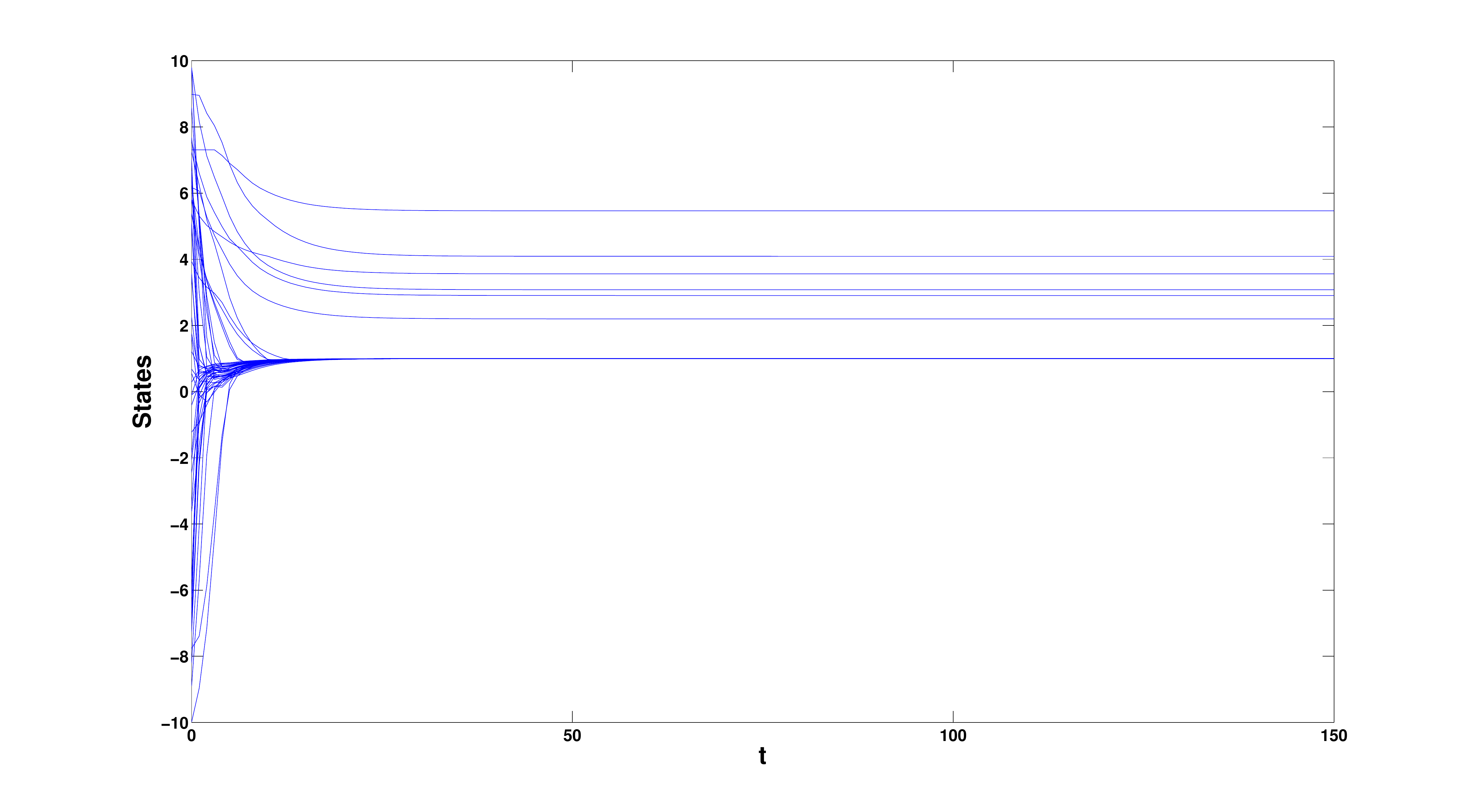}}
\caption{\label{Fig:fixed_homo} Homogeneous agents with fixed graph.}
\end{center}
\end{figure}

\begin{figure}[!tb]
\begin{center}
\subfloat[$\frac{1}{N} \sum_{i=1}^{N} x_{i} (0) = -0.9821$]{\includegraphics[width=0.75\textwidth]{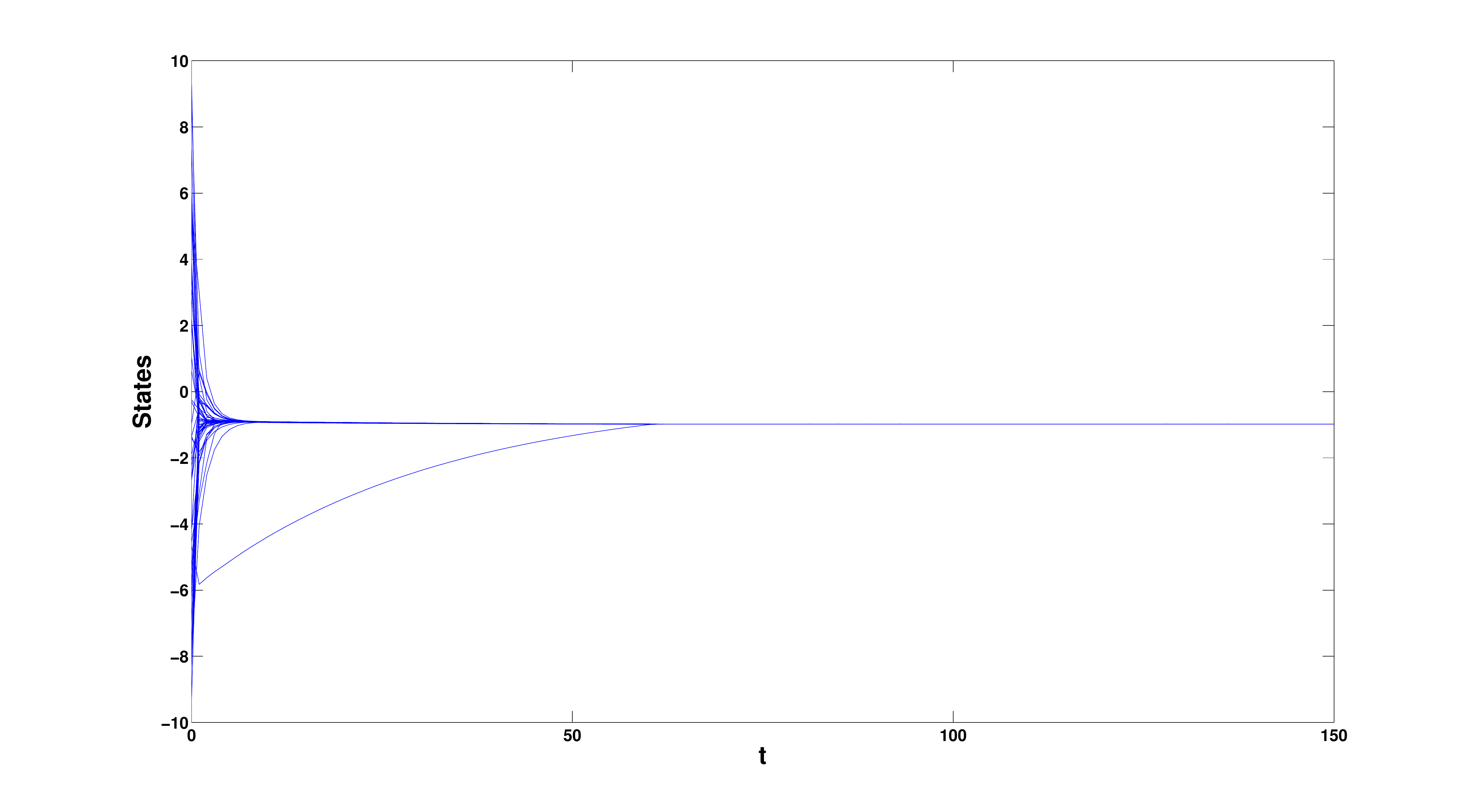}}\\
\subfloat[$\frac{1}{N} \sum_{i=1}^{N} x_{i} (0) = 1.3060$]{\includegraphics[width=0.75\textwidth]{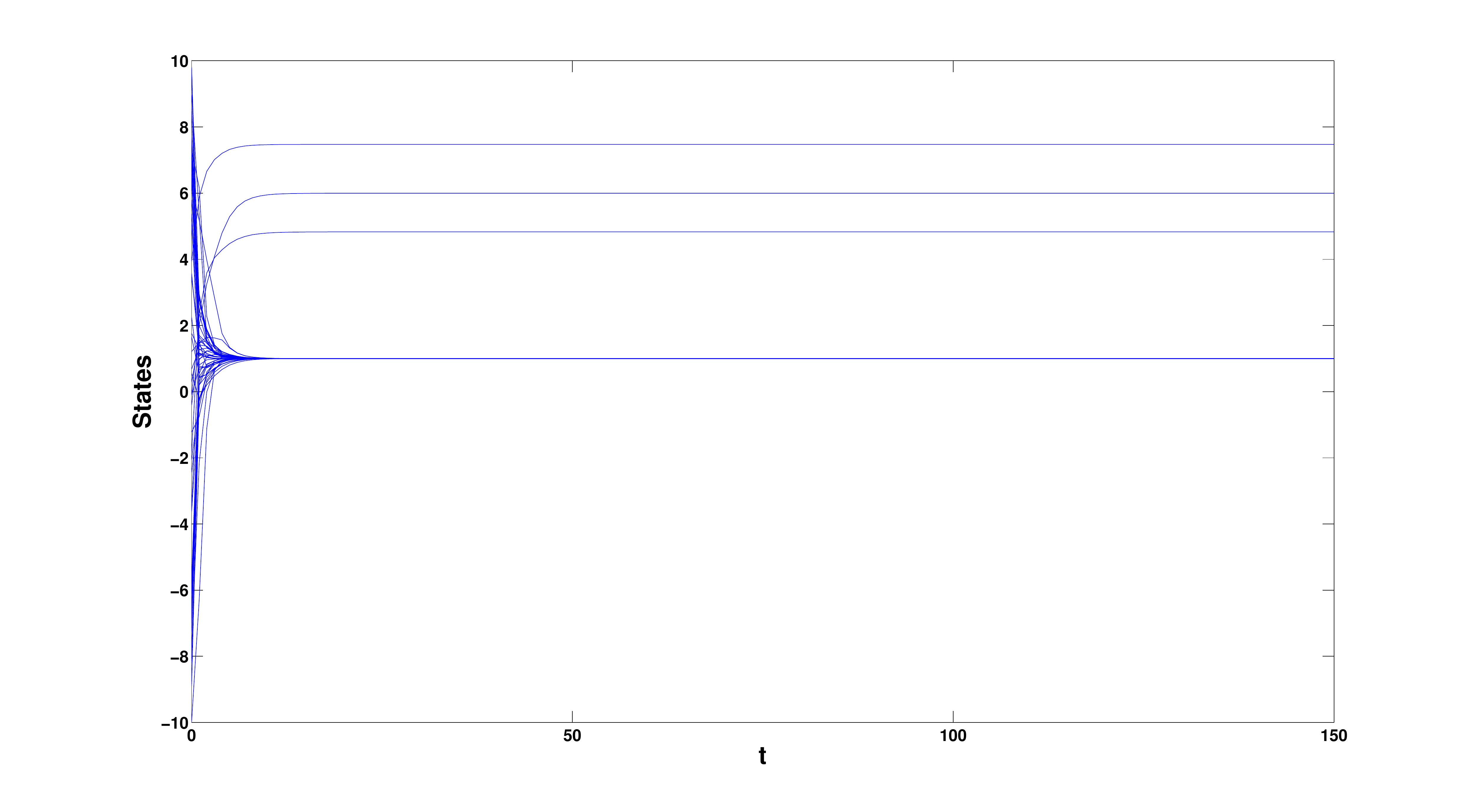}}
\caption{\label{Fig:fixed_hetero} Heterogeneous agents with fixed graph.}
\end{center}
\end{figure}

\subsection{Time-Varying Graph} \label{ch4_sim_time-varying}

We consider a group of $4$ agents whose graph topology is time-varying.
We assume that the network is changed between three graphs in Fig.~\ref{Fig_ch4_sim_time-varying_graphs} over a sequence $0 = t_{0} < t_{1} < \cdots < t_{k} < \cdots$
with $t_{k+1} - t_{k} = 10 (s)$, and $\delta_{1} = 3$, $\delta_{2} = 6$.
Note that this network is disconnected all the time.

Then, we first consider the homogeneous agents with $s = 1$.
Fig.~\ref{Fig:varying_homo} shows the simulation results with the average values as (a) $-0.75$ and (b) $1.25$.
Since the graph is integrally connected over $[0, \infty)$, it is clear that from \textit{Theorem \ref{ch4_the_time-varying_homo}}, the case (a) achieves the consensus, but the case (b) is not.

We next consider the heterogeneous agents with $s_{i} = i$ $\forall i \in \mathcal{V}:= \{1,2,3,4\}$.
With the same initial conditions as used in the homogeneous case, the simulation result is given in Fig.~\ref{Fig:varying_hetero}.
From \textit{Theorem \ref{ch4_the_time-varying_hetero}}, it is clear that the case (a) achieves the consensus, but the case (b) is not.
Moreover, from Section~\ref{subsec_unach}, the agents except for $1$ agent converge to $1$.

\begin{figure}[!tb]
\centering
\subfloat[$t \in [ t_{k}, t_{k} + \delta_{1})$]{%
\begin{tikzpicture}[-,>=stealth',shorten >=1pt,auto,node distance=1.5cm, thick,main node/.style={circle,fill=white!20,draw,font=\sffamily\small}]
  \node[main node] (1) {1};
  \node[main node] (2) [below of=1] {2};
  \node[main node] (3) [right of=2] {3};
  \node[main node] (4) [right of=1] {4};

  \path[every node/.style={font=\sffamily\small}]
    (1) edge node[right] {$3 + \sin (t)$} (2)
    ;
\end{tikzpicture}
}
\quad
\subfloat[$t \in [ t_{k} + \delta_{1}, t_{k} + \delta_{2})$]{%
\begin{tikzpicture}[-,>=stealth',shorten >=1pt,auto,node distance=1.5cm, thick,main node/.style={circle,fill=white!20,draw,font=\sffamily\small}]
  \node[main node] (1) {1};
  \node[main node] (2) [below of=1] {2};
  \node[main node] (3) [right of=2] {3};
  \node[main node] (4) [right of=1] {4};

  \path[every node/.style={font=\sffamily\small}]
    (1) edge node[right] {$2 - \cos (t)$} (3)
    ;
\end{tikzpicture}
}
\quad
\subfloat[$t \in [ t_{k} + \delta_{2}, t_{k + 1})$]{%
\begin{tikzpicture}[-,>=stealth',shorten >=1pt,auto,node distance=1.5cm, thick,main node/.style={circle,fill=white!20,draw,font=\sffamily\small}]
  \node[main node] (1) {1};
  \node[main node] (2) [below of=1] {2};
  \node[main node] (3) [right of=2] {3};
  \node[main node] (4) [right of=1] {4};

  \path[every node/.style={font=\sffamily\small}]
    (2) edge node[right] {$1.5 - \sin (t)$} (4)
    ;
\end{tikzpicture}
}
\caption{Three graphs in \textit{Section~\ref{ch4_sim_time-varying}}} \label{Fig_ch4_sim_time-varying_graphs}
\end{figure}
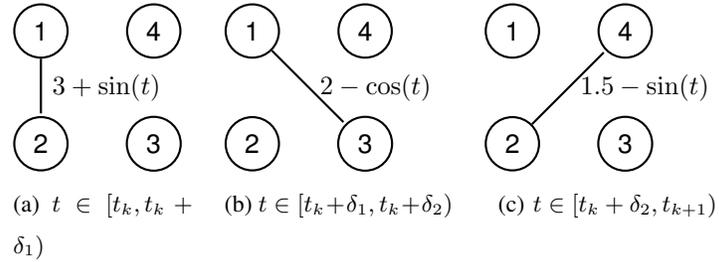

\begin{figure}[!tb]
\begin{center}
\subfloat[$\frac{1}{N} \sum_{i=1}^{N} x_{i} (0) = -0.75$]{\includegraphics[width=0.75\textwidth]{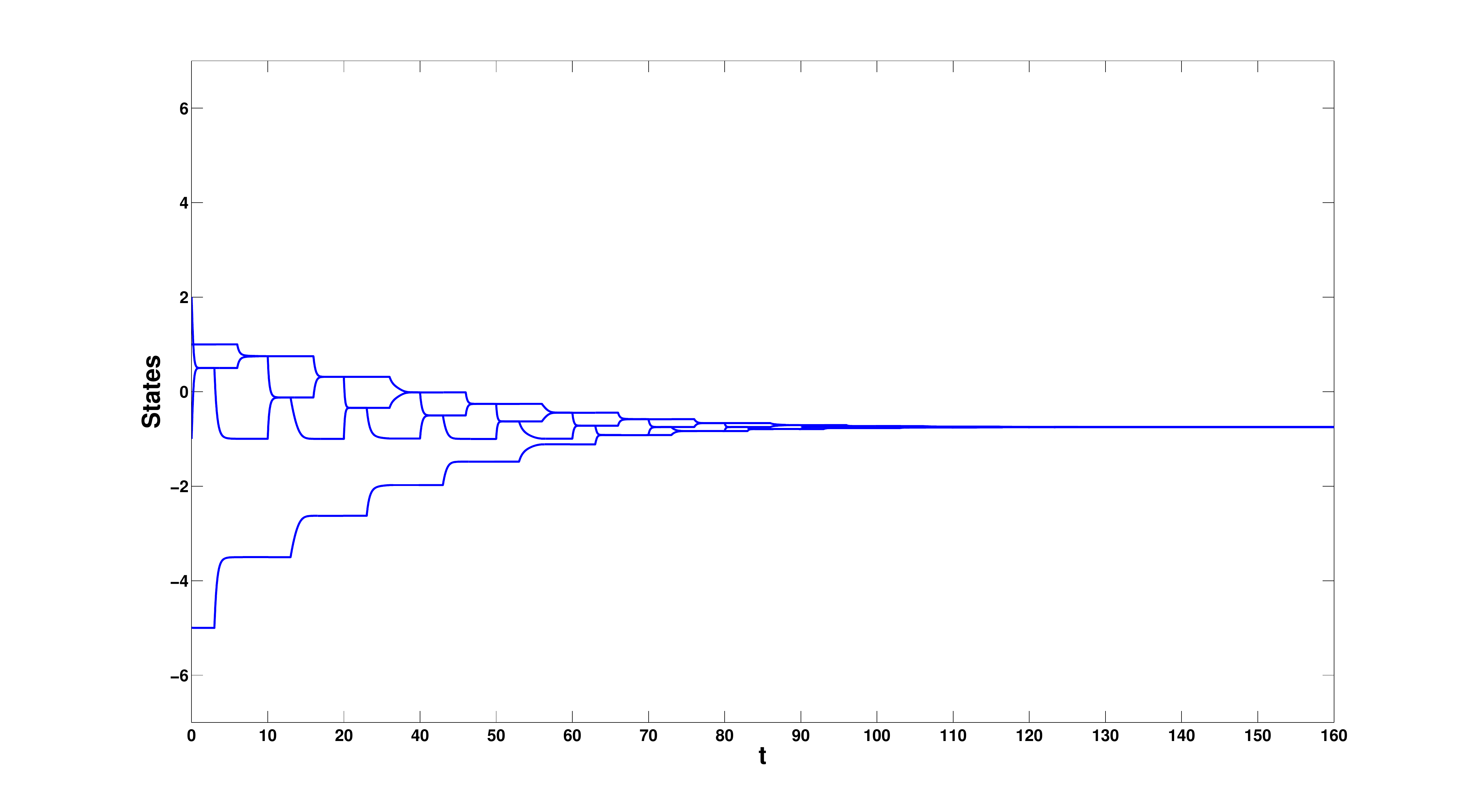}}\\
\subfloat[$\frac{1}{N} \sum_{i=1}^{N} x_{i} (0) = 1.25$]{\includegraphics[width=0.75\textwidth]{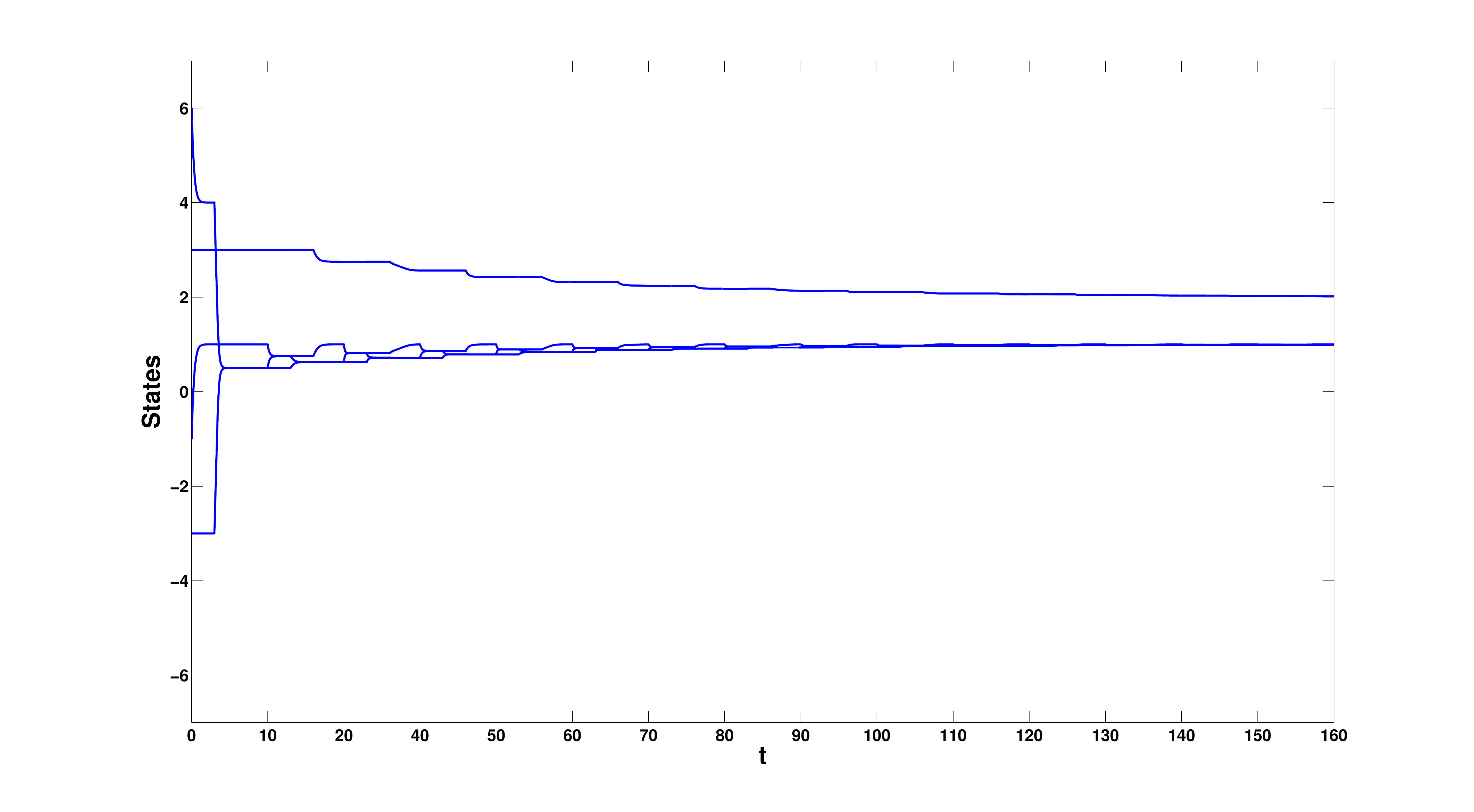}}
\caption{\label{Fig:varying_homo} Homogeneous agents with time-varying graph.}
\end{center}
\end{figure}

\begin{figure}[!tb]
\begin{center}
\subfloat[$\frac{1}{N} \sum_{i=1}^{N} x_{i} (0) = -0.75$]{\includegraphics[width=0.75\textwidth]{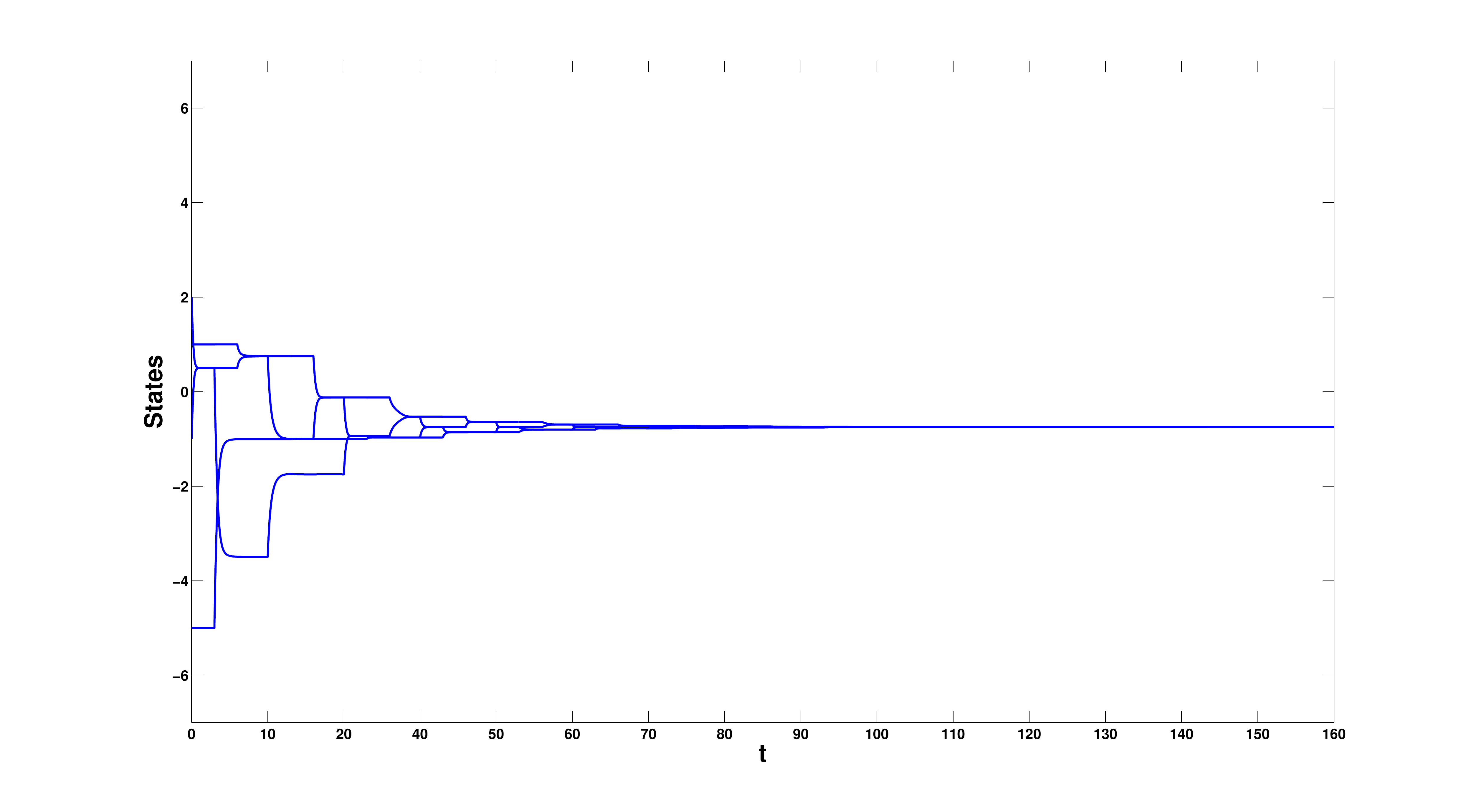}}\\
\subfloat[$\frac{1}{N} \sum_{i=1}^{N} x_{i} (0) = 1.25$]{\includegraphics[width=0.75\textwidth]{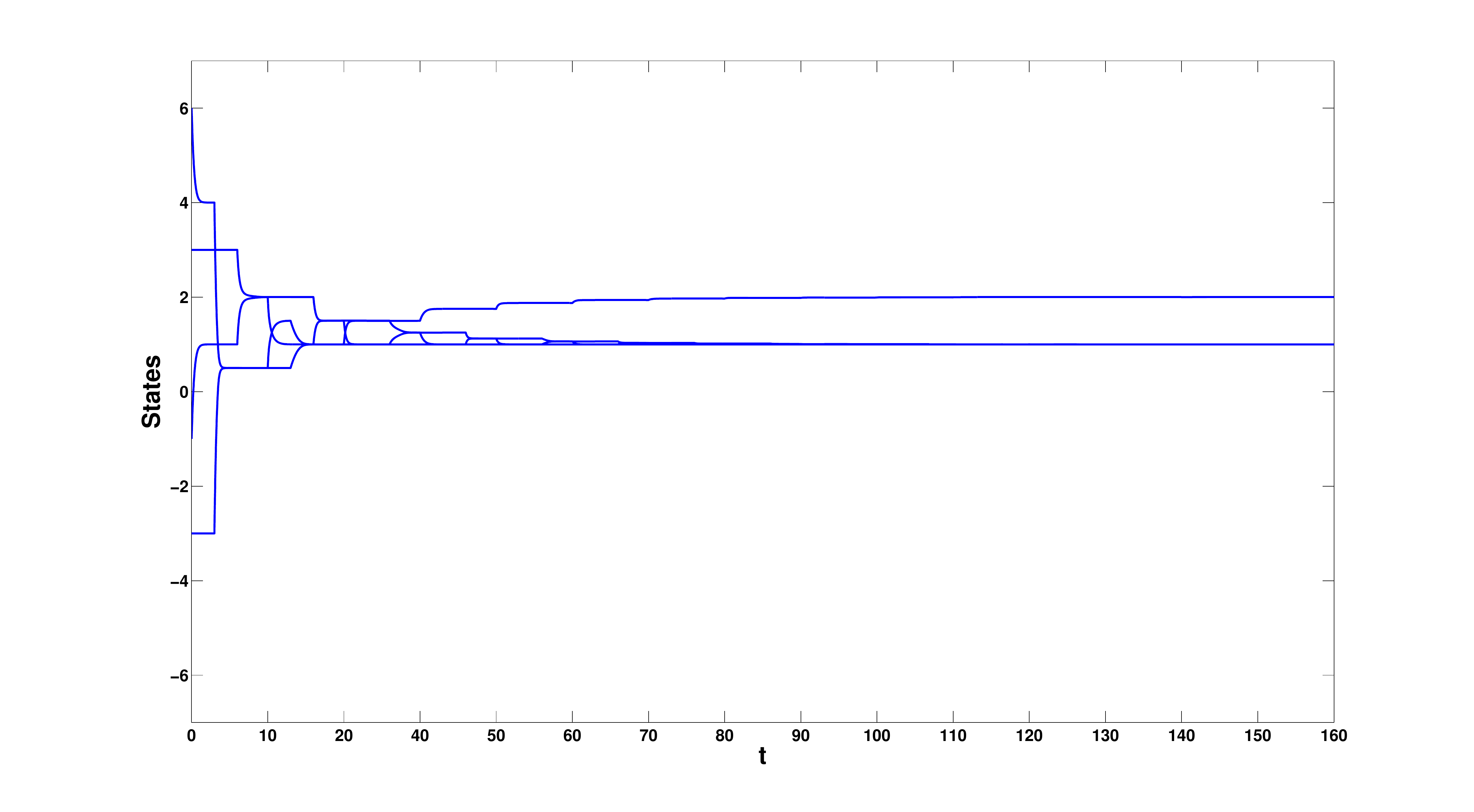}}
\caption{\label{Fig:varying_hetero} Heterogeneous agents with time-varying graph.}
\end{center}
\end{figure}


\subsection{Double-Integrator}

We consider a group of $10$ double-integrator modeled agents whose topology is fixed, undirected and connected, and the homogeneous saturation level with $s = 1$.
The initial conditions are uniformly distributed on the interval $[-10,10]$.
Fig.~\ref{Fig:double_1} and Fig.~\ref{Fig:double_2} show the simulation results with the average of all velocities are (a) $-0.85$ and (b) $1.85$, respectively.
As we can see from the simulation results, the agents achieve the consensus for the case (a), but not for the case (b).

\begin{figure}[!tb]
\begin{center}
\includegraphics[width=0.75\textwidth]{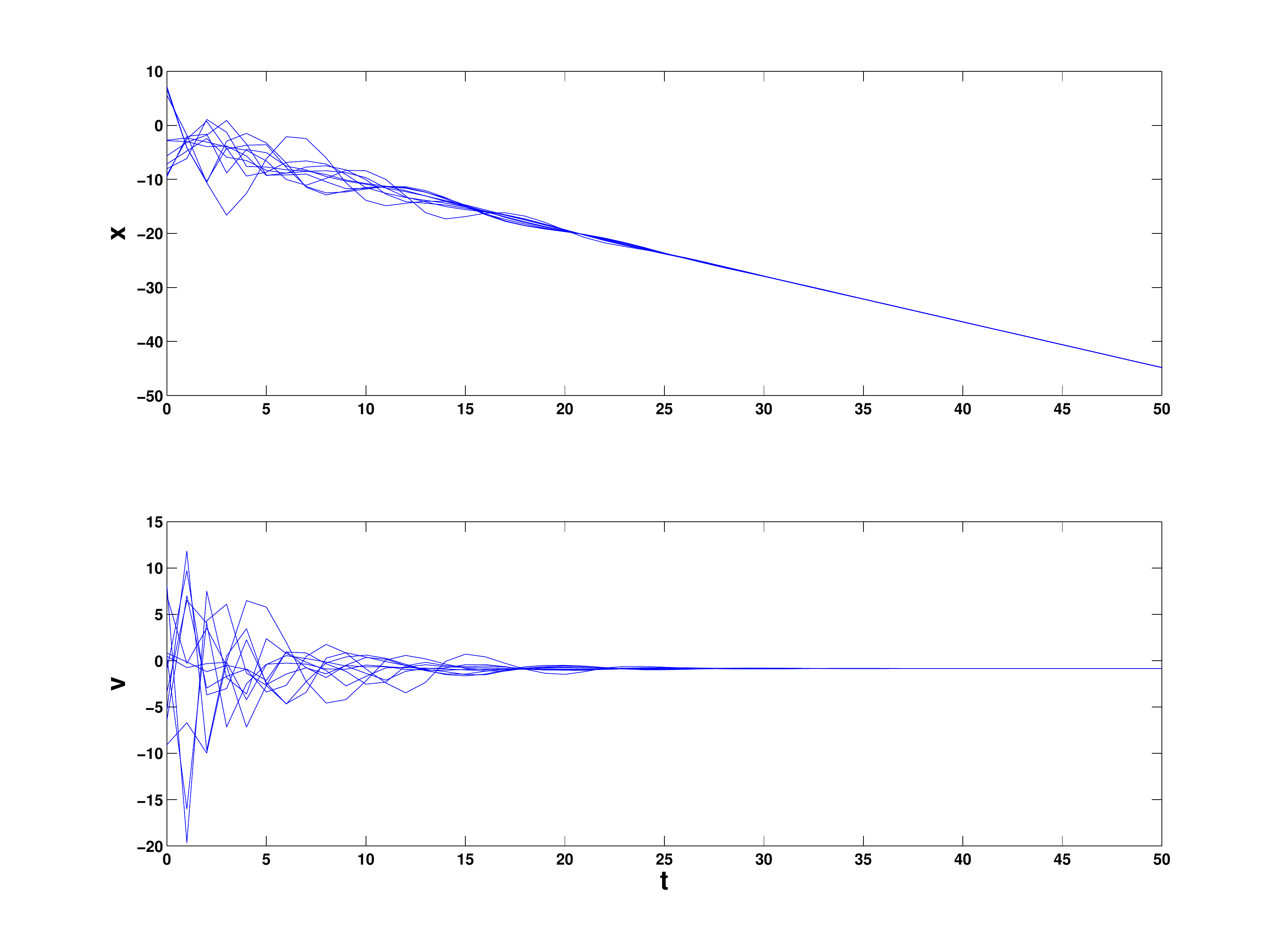}
\caption{\label{Fig:double_1}Double-integrators with $\frac{1}{N} \sum_{i=1}^{N} v_{i} (0) = - 0.85$.}
\end{center}
\end{figure}

\begin{figure}[!tb]
\begin{center}
\includegraphics[width=0.75\textwidth]{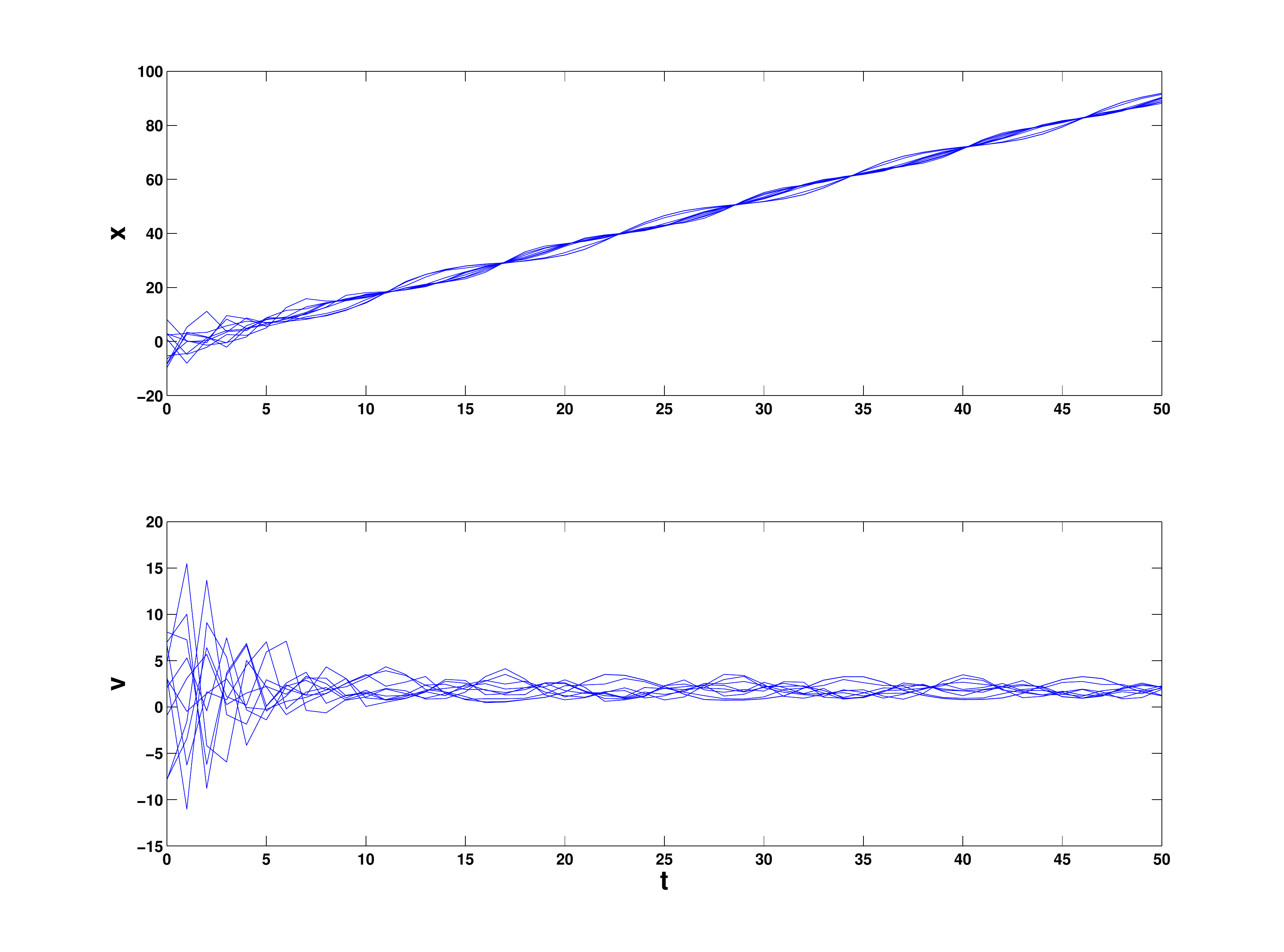}
\caption{\label{Fig:double_2}Double-integrators with $\frac{1}{N} \sum_{i=1}^{N} v_{i} (0) = 1.85$.}
\end{center}
\end{figure}

\subsection{Directed graph} \label{subsec_sim_directed}

We consider a group of $6$ agents whose graph topology is fixed, directed and strongly connected as depicted in Fig.~\ref{Fig_subsec_sim_directed_topology}, and the homogeneous saturation level with $s = 1$.
From the Laplacian matrix as in Fig.~\ref{Fig_subsec_sim_directed_topology} (b), its left eigenvector is given by
$p = [0.0678,0.0339,0.2373,0.1186,0.2712,0.2712]^{T}$.
Fig.~\ref{Fig:fixed_directed} shows simulation results with the weighted averages, $\sum_{i=1}^{N} p_{i} x_{i} (0)$, as (a) $0.3455$ and (b) $-1.8450$, respectively.
From \textit{Theorem \ref{the_directed}}, it is clear that the case (a) achieves the consensus, but the case (b) is not.



\begin{figure}[!tb]
\centering
\subfloat[Graph topology]{
\begin{tikzpicture}[->,>=stealth',shorten >=1pt,auto,node distance=1.5cm, thick,main node/.style={circle,fill=white!20,draw,font=\sffamily\scriptsize}]
  \node[main node] (1) {1};
  \node[main node] (2) [right of=1] {2};
  \node[main node] (3) [below right of=2] {3};
  \node[main node] (4) [below left of=3] {4};
  \node[main node] (5) [left of=4] {5};
  \node[main node] (6) [above left of=5] {6};
    
  \path[every node/.style={font=\sffamily\small}]
    (1) edge node[above] {} (2)
    (1) edge node[above right] {} (3)  
    (2) edge node[below right] {} (5)
    (5) edge node[above right] {} (1)
    (3) edge node[below right] {} (4)
    (4) edge node[below] {} (5) 
    (5) edge node[below left] {} (6)
    (6) edge node[above left] {} (3)                   ;
        
\end{tikzpicture}
}
\quad
\subfloat[Laplacian]{
\begin{tabular}{ll}
$L = \left[ 
\begin{array}{*{20}{c}} 
4 & -1 & -3 & 0 & 0 & 0 \\
0 & 2 & 0 & 0 & -2 & 0 \\
0 & 0 & 2 & -2 & 0 & 0 \\
0 & 0 & 0 & 4 & -4 & 0 \\
-1 & 0 & 0 & 0 & 2 & -1 \\
0 & 0 & -1 & 0 & 0 & 1
\end{array} \right]$
\end{tabular}
}
\caption{Graph topology and its Laplacian matrix in \textit{Section~\ref{subsec_sim_directed}}} \label{Fig_subsec_sim_directed_topology}
\end{figure}

\begin{figure}[!tb]
\begin{center}
\subfloat[$\sum_{i=1}^{N} p_{i} x_{i} (0) = 0.3455$]{\includegraphics[width=0.75\textwidth]{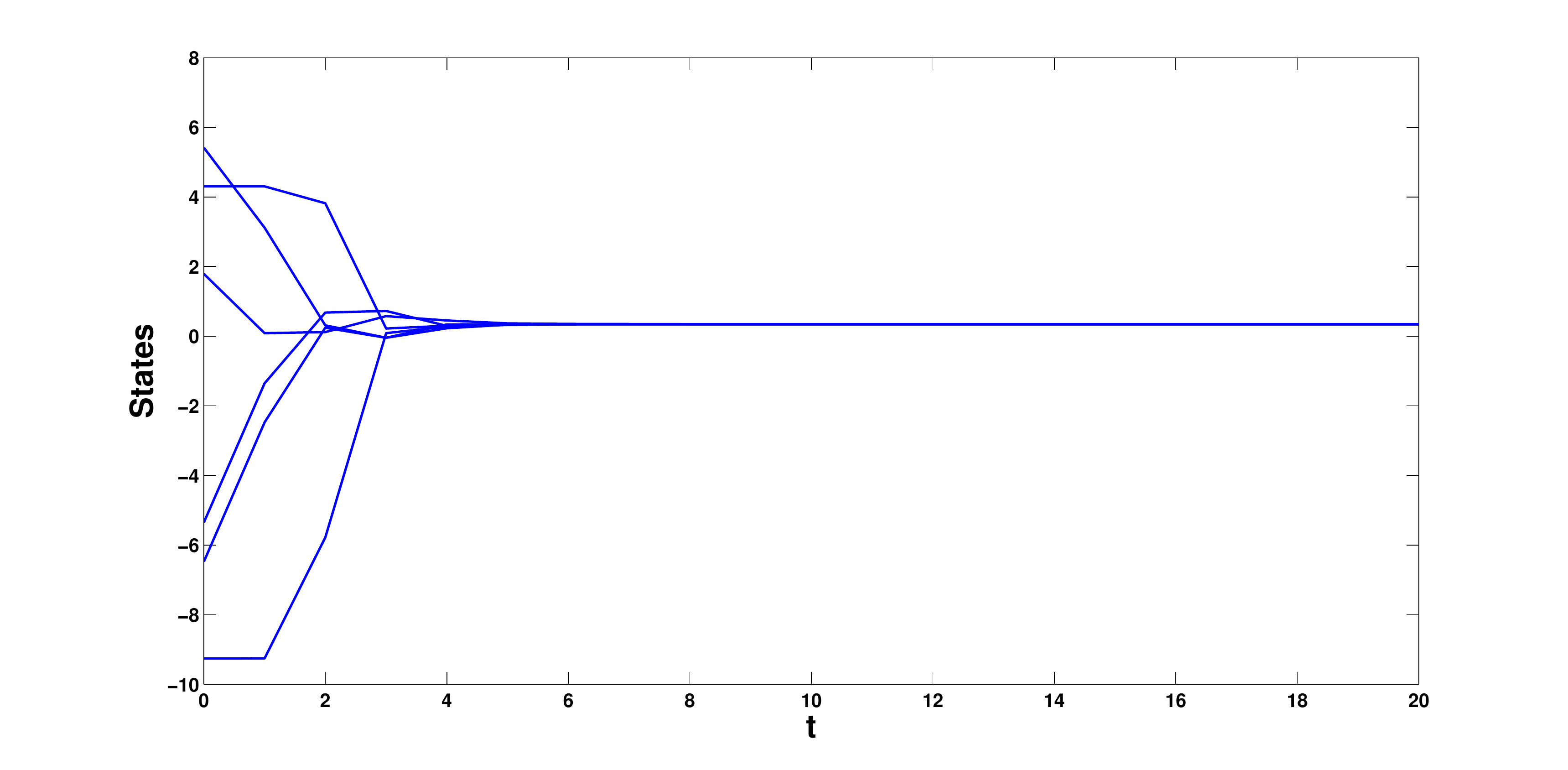}}\\
\subfloat[$\sum_{i=1}^{N} p_{i} x_{i} (0) = -1.8450$]{\includegraphics[width=0.75\textwidth]{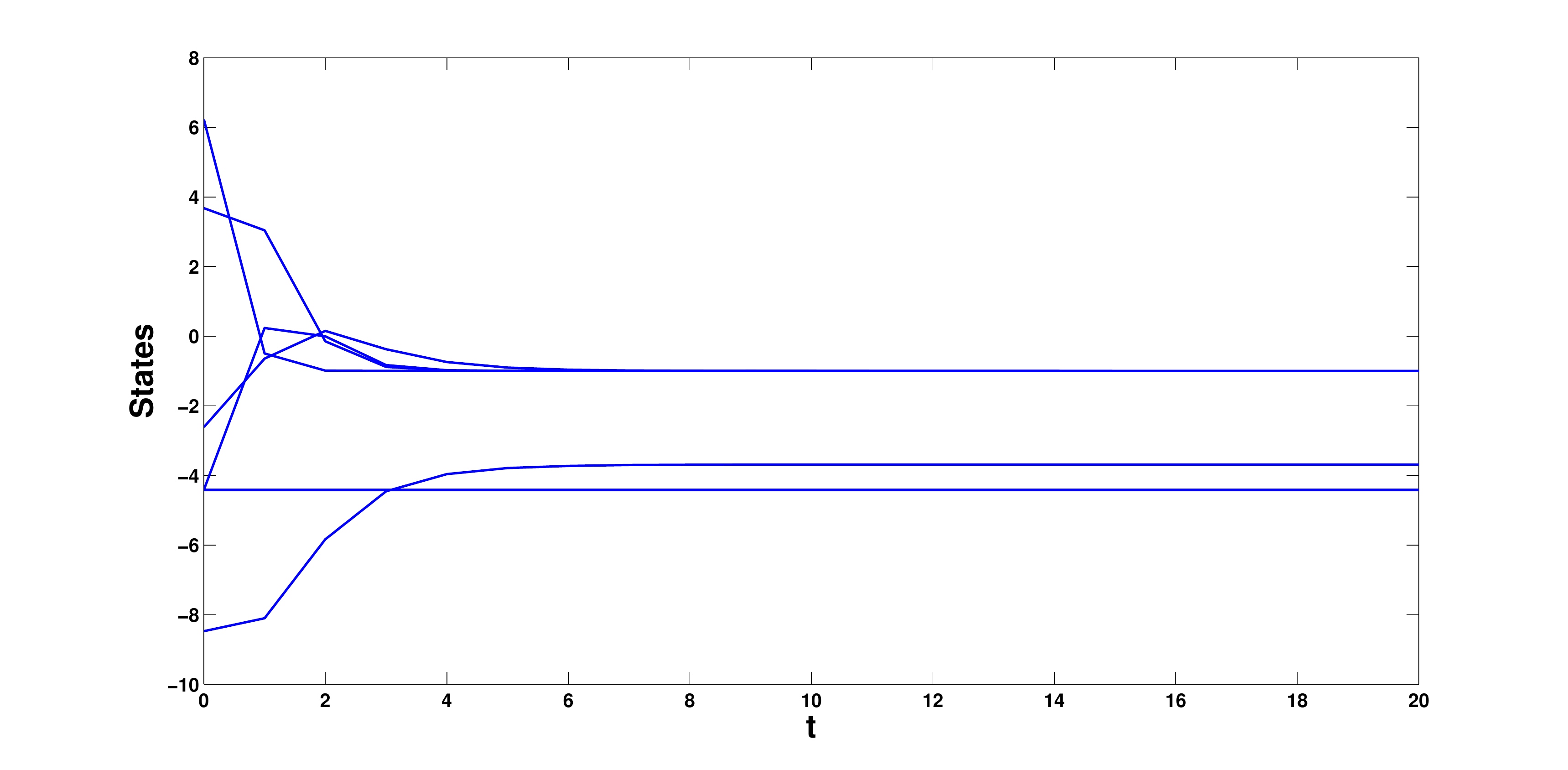}}
\caption{\label{Fig:fixed_directed} Heterogeneous agents with directed graph.}
\end{center}
\end{figure}


\section{Conclusions and Future Work} \label{conclusion}
In this paper, we have studied the consensus problem with output saturations.
Due to the existence of unachievable equilibrium for the consensus, the agents can not achieve the global consensus in the presence of output saturations.
Therefore, we have investigated the conditions for achieving the consensus, that is the exact domain of attraction.
We have discussed both homogeneous and heterogeneous saturation levels, and fixed and time-varying graphs.
To find the consensus conditions, we have analyzed the attractivity of equilibrium.
Then, by investigating the equilibrium, the necessary and sufficient conditions for achieving the consensus have been derived.

There are some issues, not addressed in this paper:
1) in Section~\ref{subsec_directed}, we have dealt with the fixed and directed graph.
Due to the invariance of the weighted average $\sum_{i=1}^{N} p_{i} x_{i} (t)$, we have proved the consensus by extending the result of the undirected graph.
However, for the time-varying directed graph, the weighted average is not invariant,
and this problem appears quite challenging in a technical sense.
2) as mentioned in \textit{Remark \ref{rem_nonlinearity}}, the analysis of this paper can be applied to any bounded nonlinearities, which are strictly increasing within the bounds.
However, the nonlinearities should be componentwise with respect to the state vector.
It would be worthwhile to extend the results of this paper to general multi-dimensional systems for real applications.
3) This paper does not address the speed of convergence.
Due to the existence of saturations, the state trajectories become nonlinear outside the saturation limits,
which makes the problem challenging.

\appendix%

\section*{Proof of Lemma \ref{ch4_lem_time-varying_hetero_conv_N-1}} 


The proof of \textit{Lemma \ref{ch4_lem_time-varying_hetero_conv_N-1}} is outlined as follows.
We first show that for any $x_{i} (t_{0}) \in \mathbf{R}$ $\forall i \in \mathcal{V}$, $x_{1} (t)$ will converge to its linear region in finite time and remain in it, that is $\exists T_{1} \ge t_{0}$ such that $|x_{1} (t) | \le s_{1}$, $\forall t \ge T_{1}$.
We next show that $|x_{1} (t) |$ for $t \ge T_{1}$ will be converge to $s_{2}$ faster than exponential.
By repeating this process for $\forall i \in \mathcal{V}_{N-1}$, we will prove \textit{Lemma \ref{ch4_lem_time-varying_hetero_conv_N-1}}.
To complete this process, we need the following lemmas.

\begin{lemma} \label{ch4_lem_time-varying_hetero_invariant}
If $x_{i} (t^{*}) \in [ - s_{k} , s_{k} ]$, $\forall i \in \mathcal{V}_{k} := \{ 1,2,...,k \}$, $t^{*} \ge t_{0}$, then $x_{i} (t) \in [ - s_{k} , s_{k} ]$, $\forall i \in \mathcal{V}_{k}$ and $\forall t \ge t^{*}$. 
\end{lemma}
\begin{proof}
Let $V_{M} ( x(t)) = \max_{i \in \mathcal{V}_{k}} \{ x_{i} (t) \}$.
Then, we will show that $D^{+} V_{M} ( x(t)) \le 0$ when $V_{M} ( x(t)) = s_{k}$.
Let $\mathcal{I} (t) = \{ i \in \mathcal{V}_{k} : x_{i} (t) = \max_{i \in \mathcal{V}_{k}} \{ x_{i} (t) \} \}$ be the index set where the maximum is reached at $t$, and consider the upper Dini derivative of $V_{M}$ as follows:
\begin{align}
D^{+} V_{M} ( x(t)) = \max_{i \in \mathcal{I} (t)} \dot{x}_{i} = \max_{i \in \mathcal{I}(t)} \sum_{j=1}^{N} \alpha_{ij} (t) ( y_{j} - y_{i} ).
\end{align}
Then, for $V_{M} = s_{k}$ and $t \ge t^{*}$, it follows that
\begin{align}
D^{+} V_{M}  =& \max_{i \in \mathcal{I}(t)} \left( \sum_{j=1}^{k} \alpha_{ij} (t) ( x_{j} - x_{i} )   + \sum_{k+1}^{N} \alpha_{ij} (t) ( y_{j} - x_{i} ) \right) \nonumber\\
\le& 0.
\end{align}
We next define $V_{m} (x(t)) = \min_{i \in \mathcal{V}_{k}} \{ x_{i} (t) \}$.
Then, we can similarly show that for $V_{m} = - s_{k}$, $D^{+} V_{m} \ge 0$, which completes the proof.
\end{proof}

We next consider the following group of $N$ agents:
\begin{align} \label{ch4_sys_time-varying_diagonal}
\dot{x}_{i} = \sum_{j=1}^{N} \alpha_{ij} (t) ( x_{j} - x_{i} ) - d_{i} (t) x_{i},
\end{align}
where $d_{i} (t)$ is continuous except for a set with measure zero, and satisfies $d_{i} (t) \ge 0$, $\forall t \ge t_{0}$ $\forall i \in \mathcal{V}$.

\begin{definition} \label{ch4_def_diagonal_exponential}
The agents (\ref{ch4_sys_time-varying_diagonal}) is said to be exponentially converge to its equilibrium $x^{*}$ with respect to $k$ if there exist two constants $\Delta, \delta > 0$ such that $|| x (t_{k}) - x^{*} || \le \Delta e^{- \delta k} || x ( t_{0}) - x^{*}  ||$
\end{definition}

\begin{lemma} \label{ch4_lem_time-varying_diagonal}
Suppose that the graph $\mathcal{G} (t)$ is integrally connected with \textit{Assumption \ref{ch4_assumption_hetero}}.
Then,\\
1) the agents (\ref{ch4_sys_time-varying_diagonal}) exponentially achieve the consensus with respect to $k$.\\
2) if there exists at least one agent such that $\int_{t_{0}}^{\infty} d_{i} (t) = \infty$, then the equilibrium point is given by the origin.
\end{lemma}
\begin{proof}
Since $d_{i} (t) \ge 0$, $\forall t \ge t_{0}$, 
the proof of the condition 1) directly follows from \textit{Theorem 5.2} in \cite{Shi:2013}.
We next prove the condition 2).
Let $S(t) = \sum_{i=1}^{N} x_{i} (t)$ and then the derivative of $S(t)$ is given by
\begin{align}
D^{+} S(t) 
=& \sum_{i=1}^{N} \sum_{j=1}^{N} \left( \alpha_{ij} (t) ( x_{j} (t) - x_{i} (t) ) - d_{i} (t) x_{i} (t) \right) \nonumber\\
=& - \sum_{i=1}^{N} d_{i} (t) x_{i} (t),
\end{align}
and thus the solution $S(t)$ is
\begin{align}
S(t) = S(t_{0}) - \int_{t_{0}}^{t} \sum_{i=1}^{N} d_{i} (\tau) x_{i} ( \tau ) d \tau.
\end{align}
From the condition 1), we know that the agents (\ref{ch4_sys_time-varying_diagonal}) achieve the consensus, and thus $S(t)$ converges to some $S^{*}$.
Therefore, it follows that
\begin{align} \label{lem_time-varying_diagonal_proof_eq}
\int_{t_{0}}^{\infty} \sum_{i=1}^{N} d_{i} ( \tau ) | x_{i} (\tau ) | d \tau = | S ( t_{0}) - S^{*} | < \infty.
\end{align}
If there exists at least one agent such that $\int_{t_{0}}^{\infty} d_{i} ( \tau) d \tau  = \infty$, 
(\ref{lem_time-varying_diagonal_proof_eq}) implies that $|x_{i} (t) |$ must converge to the origin, which completes the proof.
\end{proof}

Then, now we are going to prove \textit{Lemma \ref{ch4_lem_time-varying_hetero_conv_N-1}}.



\begin{proofof} \textit{Lemma \ref{ch4_lem_time-varying_hetero_conv_N-1}}:

Step 1.
As mentioned above, we will first show that, for any $x_{i} (t_{0}) \in \mathbf{R}$ $\forall i \in \mathcal{V}$, $x_{1} (t)$ will enter the interval $(-s_{1}, s_{1} )$ in finite time, and remain in it.

Consider the time derivative of $|x_{1} (t)|$ as follows:
\begin{align}
D^{+} | x_{1} (t) | \le& \sum_{j=1}^{N} \alpha_{1j} (t) ( | y_{j} (t) | - | y_{1} (t) | ) \nonumber\\
\le& \sum_{j=1}^{N} \alpha_{1j} (t) ( s_{2} - | y_{1} (t) | ).
\end{align}
Then, the solution is given by
\begin{align}
| x_{1} (t) | \le | x_{1} (t_{0}) | + \int_{t_{0}}^{t} \sum_{j=1}^{N} \alpha_{1j} ( \tau ) ( s_{2} - y_{1} ( \tau ) ) d \tau.
\end{align}
If $|x_{1} (t) | \ge s_{1}$, $\forall t \ge t_{0}$, it follows that
\begin{align}
| x_{1} (t) | \le& | x_{1} (t_{0}) | + \int_{t_{0}}^{t} \sum_{j=1}^{N} \alpha_{1j} ( \tau ) ( s_{2} - s_{1} ) d \tau \nonumber\\
\le& | x_{1} (t_{0}) | - s_{1,2} \int_{t_{0}}^{t} \sum_{j=1}^{N} \alpha_{1j} (\tau ) d \tau,
\end{align}
where $s_{i,j} = s_{j} - s_{i}$.
Since the graph $\mathcal{G} (t)$ is integrally connected, i.e., $\int_{t_{0}}^{\infty} \sum_{j=1}^{N} \alpha_{1j} ( \tau) d \tau = \infty$, it follows that $\lim_{t \rightarrow \infty} | x_{1} (t) | = - \infty$, which contradicts $|x_{1} (t) | \ge 0$, $\forall t \ge t_{0}$.
Moreover, from \textit{Lemma \ref{ch4_lem_time-varying_hetero_invariant}}, we can conclude that there exists $T > 0$ such that for any $|x_{1} (t_{0})| \ge s_{1}$ and $\forall t \ge T$, it holds $x_{1} (t) \in ( - s_{1} , s_{1} )$.
Moreover, since the consensus algorithm is bounded and the average value is invariant, the remaining states remain bounded for any finite time (see, \cite{Marchand:2005}).

Step $p$, $p = 2,...,N-1$.
In this step, we will show that, for any $x_{p} (t_{0} ) \in \mathbf{R}$ and $p = 2,...,N-1$, $x_{p} (t)$ will enter the interval $(-s_{p}, s_{p} )$ in finite time, and remains in it.

In the previous step, we have shown that $\forall i \in \mathcal{V}_{p-1}$, $x_{i}$ will enter and remain in the interval $(-s_{i}, s_{i})$ in finite time. Thus, the resulting dynamics of agent $i$ for $i \in \mathcal{V}_{p-1}$ is given by
\begin{align}
\dot{x}_{i} (t) = \sum_{j=1}^{p-1} \alpha_{ij} (t) ( x_{j} (t) - x_{i} (t) ) + \sum_{j=p}^{N} \alpha_{ij} (t) ( y_{j} (t) - x_{i} (t) ).
\end{align}
We next consider the upper Dini derivative of $|x_{i} (t) |$ for $i \in \mathcal{V}_{p-1}$ as follows:
\begin{align}
D^{+} | x_{i} (t) | \le& \sum_{j=1}^{p-1} \alpha_{ij} (t) ( | x_{j} (t) | - | x_{i} (t) | )  + \sum_{j=p}^{N} \alpha_{ij} (t) ( | y_{j} (t) | - | x_{i} (t) | ) \nonumber\\
\le& \sum_{j=1}^{p-1} \alpha_{ij} (t) ( | x_{j} (t) | - | x_{i} ( t) | )  + \sum_{j=p}^{N} \alpha_{ij} (t) ( s_{p} - | x_{i} (t) | ).
\end{align}
Let $p(x,t) = [ |x_{1} (t) | , ..., |x_{p-1} (t) | ]^{T}$, and $L_{p-1} (t) \in \mathbf{R}^{p -1 \times p-1}$ be the Laplacian of the subgraph $\mathcal{G}_{p-1} (t) = ( \mathcal{V}_{p-1}, \mathcal{E}_{p-1} (t) , \mathcal{A}_{p-1} (t) ) \subset \mathcal{G} (t)$, and define a diagonal matrix $D_{p-1} (t) = \diag \left( \sum_{j=p}^{N} \alpha_{1j} (t) ,..., \sum_{j=p}^{N} \alpha_{p-1 j} (t) \right)$.
Then, we have
\begin{align}
D^{+} p ( x,t) \le - \left( L_{p-1} (t) + D_{p-1} (t) \right) p (x,t) + D_{p-1} (t) s_{p} \mathbf{1}.
\end{align}
We next consider the following comparison system:
\begin{align}
\dot{z} (t) = - \left( L_{p-1} (t) + D_{p-1} (t) \right) z(t) + D_{p-1} (t) s_{p} \mathbf{1},
\end{align}
where $z \in \mathbf{R}_{+}^{p-1}$.
By denoting the error vector $\bar{z} = z - s_{p} \mathbf{1}$, we have
\begin{align}
\dot{\bar{z}} (t) = - \left( L_{p-1} (t) + D_{p-1} (t) \right) \bar{z} (t).
\end{align}
Since the graph $\mathcal{G} (t)$ is integrally connected over $[0, \infty)$, without loss of generality, we assume that there are $m$ integrally connected subgraph in $\mathcal{G}_{p-1} (t)$ over $[0, \infty )$, where $p-1 \ge m \ge 1$.
Then, by rearranging the order of the nodes, the Laplacian matrix $L_{p-1} (t)$ can be written in the block matrix form as $L_{p-1} (t) = \blkdiag \left( L_{p-1}^{1} (t) , ..., L_{p-1}^{m} (t) \right)$, where $L_{p-1}^{i} (t)$ for $i = 1,...,m$ is the Laplacian matrix of the corresponding integrally connected subgraph of $\mathcal{G}_{p-1} (t)$.
We can similarly rewrite the diagonal matrix $D_{p-1} (t)$ as $D_{p-1} (t) = \blkdiag \left( D_{p-1}^{1} (t) ,..., D_{p-1}^{m} (t) \right)$ with $D_{p-1}^{i} (t) = \diag \left( d_{1}^{i} (t) ,..., d_{m_{i}}^{i} (t) \right)$.
Then, 
there exists at least one element $q \in [ 1,...,m_{i}]$ for each $i = 1,...,m$ such that $\int_{t_{0}}^{\infty} d_{q}^{i} (t) dt = \infty$.
Therefore, according to \textit{Lemma \ref{ch4_lem_time-varying_diagonal}}, $\bar{z} (t)$ converges exponentially fast to the origin with respect to $k$, that implies $z_{i} (t)$ converges exponentially fast to $s_{p}$ with respect to $k$.
Finally, according to the comparison lemma, we can conclude that, for any $|x_{i} (t_{0}) | \ge s_{p}$, $i \in \mathcal{V}_{p-1}$, $x_{i}$ will enter the interval $[-s_{p}, s_{p}]$ faster than exponential with respect to $k$, that is, there exist two constants $\Delta, \delta > 0$ such that for $i\in \mathcal{V}_{p-1}$ and $t \in [ t_{k-1} , t_{k} )$,
\begin{align}
| x_{i} (t) | \le s_{p} + \Delta (t),
\end{align}
where $\Delta (t) = \Delta e^{-\delta k}$.

To complete the proof of Step $p$, we will next prove that for any $|x_{p} (t_{k} ) | \ge s_{p}$, $x_{p}$ will converge to the interval $(-s_{p}, s_{p} )$ in finite time.
Since the graph is integrally connected, the proof is divided as the following two cases depending on the graph topology of $\bar{\mathcal{G}}_{[0,\infty)}$:

1) $\exists j \in [ p+1,..., N]$ such that $(p,j) \in \bar{\mathcal{E}}$.

Consider $|x_{p} (t) |$ and its upper Dini derivative as follows:
\begin{align}
D^{+} | x_{p} (t) | \le& \sum_{j=1}^{p-1} \alpha_{p j } (t) ( | x_{j} (t) | - | y_{p} (t) | ) + \sum_{j=p}^{N} \alpha_{pj} (t) ( | y_{j} (t) | - | y_{p} (t) | ) \nonumber\\
\le& \sum_{j=1}^{p-1} \alpha_{pj} (t) ( s_{p} + \Delta (t) - | y_{p} (t) | ) + \sum_{j=p}^{N} \alpha_{pj} (t) ( s_{p+1} - | y_{p} (t) | ).
\end{align}
We next assume that $|x_{p} (t) | \ge s_{p}$, $\forall t \ge t_{k}$.
Then, we have
\begin{align}
|x_{p} (t) |  \le  | x_{p} (t_{k} ) | 
+  \int_{t_{k}}^{t} \left(  \sum_{j=1}^{p-1} \alpha_{pj} (  \tau  ) \Delta (  \tau  )  - \sum_{j=p}^{N} \alpha_{pj} (  \tau ) s_{p,p+1}  \right)  d \tau.
\end{align}
Since $\alpha_{ij} (t)$ is upper-and lower-bounded from \textit{Assumption \ref{ch4_assumption_hetero}} and continuous over each time interval, and $\lim_{t \rightarrow \infty} \Delta (t) = 0$, it follows that $\lim_{t \rightarrow \infty} | x_{p} (t) | = - \infty$, which is a contradiction.
Therefore, for any $|x_{p} (t_{0})| \ge s_{p}$, there exists $T> 0$ such that it holds $x_{p} (t) \in ( - s_{p}, s_{p} )$, $\forall t \ge T$.

2) for  $\forall j \in [p+1,...,N]$, $(p,j) \notin \bar{\mathcal{E}}$.

In this case, there exists at least one agent $i \in \mathcal{V}_{p-1}$ such that $(i,j) \in \bar{\mathcal{E}}$, $j \in [ p+1,...,N]$.
Then, we consider the agent $i$, $i \in \mathcal{V}_{p-1}$, and its upper Dini derivative as follows:
\begin{align}
D^{+} | x_{i} (t) | \le& \sum_{j=1}^{p-1} \alpha_{ij} (t) ( | x_{j} (t) | - | x_{i} (t) | ) + \sum_{j=p}^{N} \alpha_{ij} (t) ( | y_{j} (t) | - | x_{i} (t) | ).
\end{align}
We assume that $|x_{i} (t) | \ge s_{p}$, $\forall t \ge t_{k}$.
Then, we have
\begin{align}
D^{+} | x_{i} (t) | \le \sum_{j=1}^{p} \alpha_{ij} (t) \Delta (t) - \sum_{j=p+1}^{N} \alpha_{ij} (t) s_{p,p+1},
\end{align}
which gives
\begin{align}
| x_{i} (t) |  \le | x_{i} (t_{k})|  +  \int_{t_{k}}^{t}  \left( \sum_{j=1}^{p} \alpha_{ij} ( \tau ) \Delta (  \tau  )  - \sum_{j=p+1}^{N} \alpha_{ij} ( \tau  ) s_{p,p+1}  \right) \! \! d \tau .
\end{align}
Then, following case 1), we can conclude that there exists $T' > 0$ such that it holds $x_{i} (t) \in ( - s_{p} , s_{p} )$, $\forall t \ge T'$.
Repeating this argument for every $i \in \mathcal{V}_{p}$, we can conclude that since $\bar{\mathcal{G}}_{[0, \infty)}$ is connected, there exists $T \ge T' \ge 0$ such that for any $x_{i} (t_{0}) \ge s_{p}$, $\forall i \in \mathcal{V}_{p}$, it holds $x_{i} (t) \in ( -s_{p} , s_{p} )$ $\forall t \ge T$.

Step $N$. We will show that, for any $x_{N} (t_{0}) \in \mathbf{R}$ and $|x_{i} (t) | \le s_{N-1}$, $\forall i \in \mathcal{V}_{N-1}$, $\forall t \ge T$, we have $\lim_{t \rightarrow \infty} |x_{i} | \le s_{N}$, $\forall i \in \mathcal{V}_{N-1}$.

Since we have shown in Step 1-to-(N-1) that $|x_{i} ( t )| \le s_{N-1}$, $\forall i \in \mathcal{V}_{N-1}$, $\forall t \ge T$, we assume that $|x_{i} (t_{0}) | \le s_{N-1}$, $\forall i \in \mathcal{V}_{N-1}$.
Then, for $i \in \mathcal{V}_{N-1}$, we have
\begin{align}
\dot{x}_{i} (t) = \sum_{j=1}^{N-1} \alpha_{ij} (t) ( x_{j} (t) - x_{i} (t) ) + a_{iN} (t) ( y_{N} (t) - x_{i} (t) ).
\end{align}
Then, with the same argumentation as above, we have $\lim_{t \rightarrow \infty} | x_{i} (t) | \le s_{N}$, $\forall i \in \mathcal{V}_{N-1}$, which completes the proof. 
\end{proofof}


%
%

%

\bibliographystyle{IEEEtran}
\bibliography{reference}             

\end{document}